\pdfoutput=1
\documentclass[11pt]{article}
\usepackage{amsmath,amssymb,amsthm,amscd}
\theoremstyle{definition}
\newtheorem{prop}{Proposition}[section]
\newtheorem{lemma}[prop]{Lemma}
\newtheorem{thm}[prop]{Theorem}
\newtheorem{cor}[prop]{Corollary}

\usepackage{ascmac}	
\usepackage{bm}
\usepackage[pdftex]{graphicx}
\everymath{\displaystyle}
  \makeatletter
    
    \@addtoreset{equation}{section}
  \makeatother
\usepackage{version}	
\title{Solutions to the KP hierarchy with an elliptic background}
\author{Saburo Kakei\\[2mm]
{\small Department of Mathematics, Rikkyo University,}\\
{\small 3-34-1 Nishi-ikebukuro, Toshima-ku, Tokyo 171-8501, Japan.}
}
\date{}
\begin{document} 
\maketitle
\begin{abstract}
A class of ``elliptic soliton'' solutions of 
the Kadomtsev-Petviashvili hierarchy, 
which includes a determinantal solution of Li and Zhang, 
is described in terms of pseudo-differential operator formulation. 
In our approach, the Li-Zhang solution is obtained by 
repeatedly applying the Darboux transformation to a stationary solution.
Real-valued solutions are discussed and 
various examples that display web-like patterns are presented. 
\end{abstract}
\section{Introduction}
The relationship between soliton equations and elliptic functions 
is a fascinating and deep connection. 
An example of the connection 
is elliptic function solutions of 
the Kadomtsev-Petviashvili (KP) equation, 
\begin{equation}
\left(4u_t -12uu_x-u_{xxx}\right)_x - 3u_{yy} = 0, 
\label{KPeq}
\end{equation}
where the subscripts denote the derivatives that are being taken. 
The simplest elliptic solution is given by the Weierstrass $\wp$-function
\cite{PastrasBook,AbramowitzStegun,WhittakerWatson}, 
\begin{equation}
\wp(z) = \frac{1}{z^2}+\sum_{(m,n)\neq (0,0)}\left\{
\frac{1}{\left(z+2m\omega_1+2n\omega_2\right)^2}
-\frac{1}{\left(2m\omega_1+2n\omega_2\right)^2}
\right\}, 
\label{def:wp}
\end{equation}
where $\omega_1$, $\omega_2$ is a fundamental pair of half-periods: 
$\wp(z+2\omega_i)=\wp(z)$ ($i=1,2$). 
It is well-known that \eqref{def:wp} satisfies the following differential equations:
\begin{align}
\left\{\frac{d\wp(x)}{dx}\right\}^2
& =4\wp(x)^3-g_2\wp(x)-g_3, 
\label{diffEq:WeierstrassP:1}
\\
\frac{d^2\wp(x)}{dx^2} &= 6\wp(x)^2-\frac{g_2}{2}, 
\label{diffEq:WeierstrassP:2}
\end{align}
where the constants $g_2$, $g_3$ are given by
\begin{equation}
g_2 = 60\sum_{(m,n)\neq (0,0)}\frac{1}{\left(2m\omega_1+2n\omega_2\right)^4},
\quad 
g_3 = 140\sum_{(m,n)\neq (0,0)}\frac{1}{\left(2m\omega_1+2n\omega_2\right)^6}.
\label{def:g2,g3}
\end{equation}
It is clear from \eqref{diffEq:WeierstrassP:2} that 
\begin{equation}
u(x,y,t) = -\wp\left(x\right)
\label{stationarySol_u(x,t)}
\end{equation}
is a stationary solution to the KP equation \eqref{KPeq}. 
A non-stationary solution can be attained through the application of a Galilean-type transformation: 
\begin{equation}
u(x,y,t) = -\wp\left(x+2ay+3a^2 t\right)
\label{GalileanTransformedWPsol}
\end{equation}
where $a$ is an arbitrary constant. 

Recently, Li and Zhang \cite{LiZhang2022,LiZhang2023} 
obtained Wronskian representation for a class of 
``elliptic soliton'' solutions to the Korteweg-de Vries (KdV) equation, 
\begin{equation}
4u_t =12uu_x+u_{xxx},
\label{KdVeq}
\end{equation}
and for the KP equation \eqref{KPeq}. 
The elements of Wronskians are 
a composition of Lam\'e-type plane wave factors, 
\begin{equation}
\Phi(x;k)=\frac{\sigma(x+k)}{\sigma(x)}e^{-\zeta(k)x},
\label{def:LameFcn}
\end{equation}
where $k$ is a constant. 
Here we have used the Weierstrass functions $\zeta$ and $\sigma$ 
that are related to $\wp$ as follows:
\begin{equation}
\wp(x)=-\frac{d\zeta(x)}{dx}, \quad
\zeta(x)=\frac{1}{\sigma(x)}\cdot\frac{d\sigma(x)}{dx}.
\end{equation} 
Note that \eqref{def:LameFcn} is slightly different from 
the plane wave factor used in \cite{LiZhang2022,LiZhang2023}
(up to multiplication of $\sigma(k)$). 

We remark that the term ``elliptic solitons'' is used in various ways.
In \cite{LiZhang2022,LiZhang2023}, 
as mentioned above, the solutions are constructed 
by the Weierstrass functions $\wp$, $\zeta$ and $\sigma$ and 
Li and Zhang call their solutions as ``elliptic solitons'' after 
\cite{NejhoffAtkinson,YooKongNijhoff}. 
On the other hand, 
in \cite{Krichever,Smirnov,TreibichVerdier}, 
the term ``elliptic solitons'' refers to algebro-geometric solutions 
of soliton equations that exhibit doubly-periodic behavior in one variable, 
while $\Phi(x;k)$ \eqref{def:LameFcn} is quasi-periodic in $x$:
\begin{equation}
\Phi(x+2\omega_j;k) 
=e^{2\left\{\zeta(\omega_j)k-\zeta(k)\omega_j\right\}}
\Phi(x;k) \quad (j=1,2).
\end{equation}

In this paper, we investigate ``elliptic solitons'' in former sense.
In \cite{LiZhang2022,LiZhang2023}, Li and Zhang studied a function $f$ of the form,
\begin{align}
f &= \mathcal{W}\left(\varphi_1,\ldots,\varphi_n\right), 
\label{LiZhangWronskian}\\
\varphi_j &= a_j^+ \Phi(x;k_j)e^{-\gamma(k_j)}
 + a_j^- \Phi(x;-k_j)e^{\gamma(k_j)} \quad (j=1,\ldots,n), \nonumber\\
\gamma(k) &= \wp(k)y-\frac{\wp'(k)}{2}t \quad
\left(\wp'(x)=\frac{d\wp(x)}{dx}\right), \nonumber
\end{align}
where $\mathcal{W}\left(\varphi_1,\ldots,\varphi_n\right)$ denotes 
the Wronskian determinant, 
and proved that 
\eqref{LiZhangWronskian} satisfies a Hirota-type bilinear differential equation,
\begin{equation}
\left(D_x^4-4D_xD_t-12\wp(x)D_x^2+3D_y^2\right) f \cdot f = 0.
\label{bilinearKP_wp}
\end{equation}
Here $D_x$, $D_y$, $D_t$ are Hirota's bilinear operators:
\begin{align}
&D_x^lD_y^mD_t^n f\cdot g
\nonumber\\
&=\left.
\left(\frac{\partial}{\partial x}-\frac{\partial}{\partial x'}\right)^l
\left(\frac{\partial}{\partial y}-\frac{\partial}{\partial y'}\right)^m
\left(\frac{\partial}{\partial t}-\frac{\partial}{\partial t'}\right)^n
f(x,y,t)g(x',y',t')
\right|_{x'=x,y'=y,t'=t}.
\end{align}
Although the bilinear equation \eqref{bilinearKP_wp} contains 
the $\wp$-function in the coefficient, it can be shown that 
\begin{equation}
u = -\wp(x) + \left(\log f\right)_{xx}
\label{LiZhangSolution:KP}
\end{equation}
satisfies the KP equation \eqref{KPeq} without $\wp$-function coefficient. 
In other words, $u$ of \eqref{LiZhangSolution:KP} represents a solution 
of the KP equation with the $\wp$-function background. 
In the case of the KdV equation, solutions of this class has been 
discussed in \cite{KuznetsovMikhailov} based on the approach of Shabat \cite{Shabat}.
We remark that solutions of nonlinear Schr\"odinger type equations with elliptic backgrounds
have been studied as 
a model for density-modulated quantum condensates in \cite{Takahashi}, 
and as a model for rogue waves in \cite{ChenPelinovskyWhite,FengLingTakahashi}.

The principal aim of this paper is to investigate 
hierarchy structure that includes the solution \eqref{LiZhangSolution:KP}. 
Li and Zhang \cite{LiZhang2022,LiZhang2023}
introduced higher order time variables $t_2,t_3,\ldots$ and 
studied hierarchy structure associated with the Wronskian solutions \eqref{LiZhangWronskian}. 
They obtained an infinite sequence of Hirota-type bilinear equations such as
\begin{align}
&\left(D_x^4 - 4 D_xD_{t_3}+ 3D_{t_2}^2 -g_2\right)\tilde{f}\cdot\tilde{f}=0,
\label{bilinearKP_wp:2}\\
&\left(D_x^3D_{t_2} + 2 D_{t_2}D_{t_3} - 3 D_xD_{t_4}\right)\tilde{f}\cdot\tilde{f}=0,
\label{bilinearKP_wp:deg5}\\
&\left(D_x^6-20D_x^3D_{t_3}-80D_{t_3}^2+144D_xD_{t_5}-45D_x^2D_{t_2}^2
+3g_2D_x^2-24g_3\right)\tilde{f}\cdot\tilde{f}=0,
\label{bilinearKP_wp:deg6:1}\\
&\left(D_x^6+4D_x^3D_{t_3}-32D_{t_3}^2-9D_x^2D_{t_2}^2+36D_{t_2}E_{t_4}
+3g_2D_x^2-24g_3\right)\tilde{f}\cdot\tilde{f}=0,
\label{bilinearKP_wp:deg6:2}\\
& \qquad\vdots\nonumber
\end{align}
As mentioned in \cite{LiZhang2023}, 
\eqref{bilinearKP_wp:2} is transformed to \eqref{bilinearKP_wp} 
by setting $\tilde{f}=\sigma(x)f$. 
Li and Zhang proposed a generating function of bilinear equations 
(\textit{bilinear identity}) that includes 
\eqref{bilinearKP_wp:2}--\eqref{bilinearKP_wp:deg6:2} using vertex operators
\cite{LiZhang2023}. 

If we set $g_2=g_3=0$, the bilinear equations 
\eqref{bilinearKP_wp:2}--\eqref{bilinearKP_wp:deg6:2} coincide with
those of the KP hierarchy listed in \cite{JimboMiwa}.
On the other hand, if we set
\begin{equation}
\tilde{f}=\exp\left[
\frac{g_2}{4}\left\{(1+2a)t_2^2+(1+3a)x t_3\right\}\right]\tau
\label{transformation_Q(x,t)_KPcase}
\end{equation}
for $a$ is an arbitrary constant, then $\tau$ satisfies
\begin{equation}
\left(D_x^4 - 4 D_xD_{t_3}+ 3D_{t_2}^2\right)\tau\cdot\tau=0,
\end{equation}
which is the bilinear form of the original KP equation \eqref{KPeq}. 
This suggests that the bilinear equations with elliptic solitons 
may be transformed to the equations without elliptic parameters $g_2$, $g_3$
via suitable transformation like \eqref{transformation_Q(x,t)_KPcase}. 
We will show that bilinear equations associated with 
the Li-Zhang solution \eqref{LiZhangWronskian} 
are actually equivalent to those of 
the original KP hierarchy without elliptic parameters $g_2$, $g_3$.
We will construct the transformation 
from the Li-Zhang solution \eqref{LiZhangWronskian} 
to a $\tau$-function of the KP hierarchy via pseudo-differential operator 
formulation.

It should be remarked that solutions with the $\wp$-function background has been 
discussed in \cite{NakayashikiGMP2020,Nakayashiki_arXiv2023} as certain degeneration of 
hyperelliptic solutions, using the correspondence between
algebro-geometric solutions and the Sato Grassmannian. 
We also remark that a similar results can be found in \cite{Ichikawa}. 
The approach for handling elliptic solitons in this paper doesn't rely on 
the algebro-geometric framework. 
Instead, we employ a formulation based on pseudo-differential operators 
and utilize elementary properties of the Weierstrass functions 
$\wp$, $\zeta$ and $\sigma$. 

\section{Pseudo-differential operators and KP hierarchy}
The KP hierarchy is an infinite (countable) family of 
partial differential equations that includes the KP equation \eqref{KPeq}. 
Here we give a brief review of the KP hierarchy using 
pseudo-differential operator formulation 
\cite{DJKM,OSTT,Sato,SatoSato,Sato:Bowdoin}. 

\subsection{Pseudo-differential operators}
Let $\mathcal{A}$ be an associative ring and $\partial_x:\mathcal{A}\to \mathcal{A}$ its derivation.
Let us consider the ring of formal pseudo-differential operators: 
\begin{equation}
\mathcal{E}=\left\{
\sum_{-\infty < n\ll +\infty} a_n\partial_x^n, \; a_n\in \mathcal{A}
\right\}.
\end{equation}
Multiplication can be defined according to the generalized 
Leibniz rule:
\begin{equation}
\left(\sum_{n} a_n\partial_x^n\right)\left(\sum_{m} b_m\partial_x^m\right)
=\sum_{n,m}a_i\sum_{r=0}^{\infty}\binom{n}{r} \partial_x^r(b_m) \partial_x^{n+m-r}.
\end{equation}
Here the generalized binomial coefficient is defined by
\begin{equation}
\binom{n}{r} = 
\begin{cases}
1 & (r=0),\\
\frac{n(n-1)\cdots (n-r+1)}{r!} & (r\neq 0).
\end{cases}
\end{equation}

For $P=\sum_{-\infty < n\leq m} a_n\partial_x^n
\in\mathcal{E}$ ($a_m\neq 0$), denote by $\mathrm{ord}_{\partial_x}(P)=m$ the leading order of $P$. 
Let $\mathcal{E}^{(m)}$ be a subspace of $\mathcal{E}$ with the following property:
\begin{equation}
\mathcal{E}^{(m)} = \left\{ P\in\mathcal{E} \;|\;
\mathrm{ord}_{\partial_x}(P)  \leq m\right\}, 
\end{equation}
and let $\mathcal{D}$ be a subring of differential operators: 
\begin{equation}
\mathcal{D}=\left\{
\sum_{0\leq n\ll +\infty} a_n\partial_x^n, \; a_n\in \mathcal{A}
\right\}.
\end{equation}
Then $\mathcal{E}$ has a natural decomposition 
$\mathcal{E}=\mathcal{D}\oplus\mathcal{E}^{(-1)}$. 
Denote by $(\,\cdot\,)_{\geq 0}$ the projection to $\mathcal{D}$, and 
by $(\,\cdot\,)_{<0}$ the projection to $\mathcal{E}^{(-1)}$, i.e., 
for $m>0$, 
\begin{align}
\left(\sum_{-\infty < n\leq m} a_n\partial_x^n\right)_{\geq 0}
&= \sum_{0\leq n\leq m} a_n\partial_x^n\in\mathcal{D}
,\\
\left(\sum_{-\infty < n\leq m} a_n\partial_x^n\right)_{<0}
&= \sum_{-\infty < n<0} a_n\partial_x^n\in\mathcal{E}^{(-1)}.
\end{align}

We prepare several formulas that we shall use later. 
\begin{lemma}[Helminck and van de Leur \cite{Helminck_vdLeur_CJM,Helminck_vdLeur_PublRIMS}
(cf. Oevel \cite{Oevel}, (2.10))]
\label{lemma:Helminck_vdLeur}
Let $a,b\in \mathcal{A}$, $P=\sum_{j=-\infty}^m p_j\partial^j \in\mathcal{E}$, and 
$\mathrm{Res}_{\partial_x}(P) = p_{-1}$. Then the following relations hold:
\begin{align}
& (aPb)_{\geq 0} = a (P)_{\geq 0}b,\quad
(aPb)_{<0} = a (P)_{<0}b,
\label{lemma6.1:(a),(b)}\\
& \mathrm{Res}_{\partial_x}(aPb)=ab\,\mathrm{Res}_{\partial_x}(P),
\label{lemma6.1:(c)}\\
& \left(\partial_x P\right)_{<0} = \partial_x (P)_{<0} - \mathrm{Res}_{\partial_x}(P),
\label{lemma6.1:(d)}\\
& \left(P\partial_x\right)_{<0} = (P)_{<0}\partial_x - \mathrm{Res}_{\partial_x}(P),
\label{lemma6.1:(e)}\\
& \left(P\partial_x^{-1}\right)_{<0}
= (P)_{<0}\partial_x^{-1} + \mathrm{Res}_{\partial_x}(P\partial_x^{-1})\partial_x^{-1},
\label{lemma6.1:(f)}\\
& \mathrm{Res}_{\partial_x}\left(P a\partial^{-1}\right)=(P)_{\geq 0}(a),
\label{lemma6.1:(h)}
\end{align}
where the right-hand-side of \eqref{lemma6.1:(h)} 
means that the differential operator $(P)_{\geq 0}$ acts on $a\in \mathcal{A}$.
\end{lemma}

\subsection{KP hierarchy}
We introduce an infinite set of ``time'' variables 
$\bm{t}=\left(t_2,t_3,\ldots\right)$. 
Usually $x$ is identified as $t_1$, but we do not follow 
this convention to distinguish between $x$ and the time variables $\bm{t}=\left(t_2,t_3,\ldots\right)$, 
except for Lemma \ref{lemma:Q(t)}.
We consider generalized Lax operator $L$ of the form
\begin{equation}
L = \partial_x + \sum_{j=1}^\infty u_{j+1}\partial_x^{-j}, \quad
u_{j+1}=u_{j+1}(x,\bm{t}), 
\label{def:generalizedLaxOp}
\end{equation}
We say that a Lax operator $L$ is a solution to the KP hierarchy if and only if 
it satisfies the generalized Lax equations:
\begin{equation}
\frac{\partial L}{\partial t_n} = \left[B_n,\,L\right], \quad
B_n=\left(L^n\right)_{\geq 0}
\quad (n=2,3,\ldots). 
\label{generalizedLaxEqs}
\end{equation}
The explicit forms of $B_2$, $B_3$ are given by
\begin{equation}
B_2 = \partial_x^2 + 2u_2, \quad
B_3 = \partial_x^3+3u_2+3u_3+3\frac{\partial u_2}{\partial x}.
\label{B2B3_u}
\end{equation}
The Lax equations \eqref{generalizedLaxEqs} are equivalent to the following infinite set of equations
(Zakharov-Shabat equations):
\begin{equation}
\frac{\partial B_m}{\partial t_n}-\frac{\partial B_n}{\partial t_m} = 
\left[B_n,B_m\right] \quad (m,n=2,3,\ldots).
\label{generalizedZSeqs}
\end{equation}
Applying \eqref{B2B3_u} to \eqref{generalizedZSeqs} with $m=2$, $n=3$, 
one can show that 
$u=u_2(x,t_2,t_3)$ satisfies the KP equation \eqref{KPeq} ($y=t_2$, $t=t_3$). 

The equations \eqref{generalizedLaxEqs} and \eqref{generalizedZSeqs} are 
compatibility conditions for the following auxiliary linear problem:
\begin{align}
L\psi &= \lambda \psi,\quad \psi=\psi(x,\bm{t};\lambda), 
\label{Lpsi=lambdapsi}\\
\frac{\partial\psi}{\partial t_n} &= B_n\psi \quad (n=1,2,\ldots).
\label{dnpsi=Bnpsi}
\end{align}
Let us consider a solution of \eqref{Lpsi=lambdapsi} and \eqref{dnpsi=Bnpsi} 
of the form
\begin{equation}
\begin{aligned}
\psi(x,\bm{t};\lambda) &= \hat{w}(x,\bm{t};\lambda)e^{\xi(x,\bm{t};\lambda)}, \\
\hat{w}(x,\bm{t};\lambda) 
&= 1+\sum_{j=1}^\infty w_j(x,\bm{t}) \lambda^{-j},\\
\xi(x,\bm{t};\lambda) &= 
\lambda x + \sum_{j=2}^\infty \lambda^j t_j.
\end{aligned}
\label{def:BAfcn}
\end{equation}
The function $\psi$ of the form \eqref{def:BAfcn} is called the wave function 
(or the formal Baker-Akhiezer function) 
associated with $L$. 

Define a space $M$ as \cite{Helminck_vdLeur_PublRIMS}
\begin{equation}
M = \left\{\left(\sum_{-\infty<j\ll+\infty} a_j z^j\right) e^{\xi(x,\bm{t};\lambda)}
\;\Bigg|\; a_j\in \mathcal{A}\right\},
\end{equation}
which is an $\mathcal{E}$-module: 
\begin{align}
b&\left\{\left(\sum_{-\infty<j\ll+\infty} a_j z^n\right) e^{\xi(x,\bm{t};\lambda)}\right\}
=\left(\sum_{-\infty<j\ll+\infty} ba_j z^n\right) e^{\xi(x,\bm{t};\lambda)}
\quad (b\in \mathcal{A}), \nonumber\\
\partial_x&\left\{\left(\sum_{-\infty<j\ll+\infty} a_j z^n\right) e^{\xi(x,\bm{t};\lambda)}\right\}
=\left\{\sum_j\partial_x(a_j)z^n + \sum_j a_j z^{j+1}
\right\}e^{\xi(x,\bm{t};\lambda)},
\label{action_of_partialx_on_M}\\
\frac{\partial}{\partial t_n}
&\left\{\left(\sum_{-\infty<j\ll+\infty} a_j z^n\right) e^{\xi(x,\bm{t};\lambda)}\right\}
=\left(\sum_j\frac{\partial a_j}{\partial t_n}z^n + \sum_j a_j z^{j+n}
\right)e^{\xi(x,\bm{t};\lambda)}. \nonumber
\end{align}
Using \eqref{action_of_partialx_on_M}, one can introduce 
$W\in\mathcal{E}^{(0)}$ 
as
\begin{equation}
W = 1+\sum_{j=1}^\infty w_j(x,\bm{t}) \partial_x^{-j}
\label{def:DressingOp}
\end{equation}
which is called the \textit{wave operator} or the \textit{dressing operator}. 
{}From \eqref{Lpsi=lambdapsi} and \eqref{dnpsi=Bnpsi}, we have 
\begin{align}
&\psi(x,\bm{t};\lambda) = W e^{\xi(x,\bm{t};\lambda)}, \quad
L = W\partial_x W^{-1}, \\
&\frac{\partial W}{\partial t_n} = B_n W - W\partial_x^n
=-\left( W\partial_x^n W^{-1}\right)_{<0}W.
\label{SatoWilsonEqn}
\end{align}
The equation \eqref{SatoWilsonEqn} is called the \textit{Sato-Wilson equation}. 
We remark that the coefficients of $L$ can be written in terms of $w_j$ ($j=1,2,\ldots$):
\begin{equation}
\label{u2=-w'1,u3=...,u4=...}
\begin{aligned}
u_2 &= -w'_1, \quad 
u_3 = -w'_2+w_1w'_1, \\
u_4 &= -w'_3 + w_1w'_2 + w'_1w_2 - w_1^2w'_1-\left(w'_1\right)^2, \ldots .
\end{aligned}
\end{equation}
Here and hereafter, the prime denotes differentiation with respect to $x$.

\begin{thm}[Existence of $\tau$-function \cite{DJKM}]
If $\psi(x,\bm{t};\lambda)$ is a wave function (or a formal Baker-Akhiezer function) 
of the KP hierarchy, then there exists a function $\tau(x,\bm{t})$ 
(\textit{$\tau$-funcion}) 
that is related to $\hat{w}(x,\bm{t};\lambda)$ of \eqref{def:BAfcn} as
\begin{equation}
\hat{w}(x,\bm{t};\lambda) = 
\frac{\tau\left(x-\lambda^{-1},\bm{t}-\left[\lambda^{-1}\right]\right)}{\tau(x,\bm{t})}, 
\label{hatw=tau/tau}
\end{equation}
where
\begin{equation}
\bm{t}-\left[\lambda^{-1}\right] 
= \left(t_2-\frac{1}{2\lambda^2},\,t_3-\frac{1}{3\lambda^3},\,\ldots,\,
t_n-\frac{1}{n\lambda^n},\,\ldots\right).
\end{equation}
\end{thm}
Expanding \eqref{hatw=tau/tau} with respect to $\lambda^{-1}$, one obtains (cf. \cite{OSTT})
\begin{equation}
\begin{aligned}
w_1&=-\frac{1}{\tau}\frac{\partial\tau}{\partial x},\quad
w_2=\frac{1}{2\tau}\left(\frac{\partial^2\tau}{\partial x^2}-\frac{\partial\tau}{\partial t_2}\right),\\
w_3&=-\frac{1}{6\tau}\left(\frac{\partial^3\tau}{\partial x^3}
-3\frac{\partial^2\tau}{\partial x\partial t_2}
+2\frac{\partial\tau}{\partial t_3}\right),
\;\ldots .
\end{aligned}
\end{equation}
Applying these relations to \eqref{u2=-w'1,u3=...,u4=...}, we obtain (cf. \cite{OSTT})
\begin{equation}
\label{u2=del2logtau,u3=...,u4=...}
\begin{aligned}
u_2 &= \frac{\partial^2}{\partial x^2}\log\tau,\quad 
u_3 = \frac{1}{2}\left(
-\frac{\partial^3}{\partial x^3}+
\frac{\partial^2}{\partial x\partial t_2}\right)\log\tau,\\
u_4 &= \frac{1}{6}\left(
\frac{\partial^4}{\partial x^4}
-3\frac{\partial^3}{\partial x^2\partial t_2}
+2\frac{\partial^2}{\partial x\partial t_3}
\right)\log\tau
-\left(\frac{\partial^2}{\partial x^2}\log\tau\right)^2,\;\ldots.
\end{aligned}
\end{equation}

\begin{thm}[Bilinear identity \cite{DJKM,JimboMiwa,MJDbook}]
If $\tau(x,\bm{t})$ is a $\tau$-function of the KP hierarchy, then 
$\tau(x,\bm{t})$ satisfies the following equation 
(\textit{bilinear identity}) for any $x$, $x'$, $\bm{t}$ and $\bm{t}'$:
\begin{equation}
\mathrm{Res}_\lambda\left[
e^{\xi(x-x',\bm{t}-\bm{t}';\lambda)}
\tau\left(x-\lambda^{-1},\,\bm{t}-\left[\lambda^{-1}\right]\right)
\tau\left(x'+\lambda^{-1},\,\bm{t}'+\left[\lambda^{-1}\right]\right)
\right]=0, 
\label{bilinearIdentity:KP}
\end{equation}
where $\mathrm{Res}_\lambda$ is the formal residue: 
$\mathrm{Res}_\lambda\left[\sum_n a_n\lambda^n\right]=a_{-1}$. 
Conversely, if $\tau(x,\bm{t})$ satisfies the bilinear identity \eqref{bilinearIdentity:KP}, 
then the corresponding $\psi(x,\bm{t})$ via the formulas \eqref{def:BAfcn} and 
\eqref{hatw=tau/tau} solves the auxiliary linear problem 
\eqref{Lpsi=lambdapsi}, \eqref{dnpsi=Bnpsi}. 
\end{thm}

\subsection{Galilean symmetry of the KP hierarchy}
The Galilean-type transformation used in \eqref{GalileanTransformedWPsol} 
can be generalized to the whole KP hierarchy. 
For $\bm{t}=\left(t_2,t_3,\ldots\right)$, define $T(\bm{t})$ as 
\begin{equation}
T\left(x,\bm{t}\right) = \left(
x+\sum_{j=2}^\infty\binom{j}{1}a^{j-1}t_j,\,
\sum_{j=2}^\infty\binom{j}{2}a^{j-2}t_j,\,\ldots,\,
\sum_{j=n}^\infty\binom{j}{n}a^{j-n}t_j,\,\ldots
\right). 
\end{equation}
\begin{thm}[Galilean symmetry of the KP hierarchy]
\label{thm:GalileanSymmetry}
Let $L(x,\bm{t})\in\mathcal{E}^{(1)}$ of the form \eqref{def:generalizedLaxOp}
be a solution of the KP hierarchy. 
Define $\tilde{L}(x,\bm{t})$ as
\begin{equation}
\tilde{L}(x,\bm{t}) = e^{ax}L(x,\bm{t})e^{-ax}+a
\end{equation}
for $a$ is an constant.
Then $\tilde{L}\left(x,T(\bm{t})\right)$ is also of the form \eqref{def:generalizedLaxOp}
and satisfies the KP hierarchy. 
\end{thm}
\begin{proof}
It follows from \eqref{lemma6.1:(a),(b)} that
\begin{align}
\left(\tilde{L}^m\right)_{\geq 0} &= 
\sum_{n=0}^m\binom{m}{n}a^{m-n}
e^{ax}\left(L^n\right)_{\geq 0}e^{-ax}.
\label{tildeLm}
\end{align}
We introduce a new set of time-variables 
$\tilde{\bm{t}}=\left(\tilde{t}_1,\tilde{t}_2,\ldots\right)$ and define
$\bm{t}=T\left(\tilde{\bm{t}}\right)$, i.e., 
\begin{equation}
t_n = \sum_{j=n}^\infty\binom{j}{n}a^{j-n}\tilde{t}_j \quad (n=1,2,\ldots).
\end{equation}
Using \eqref{tildeLm}, we have
\begin{align}
\frac{\partial}{\partial\tilde{t}_m}&\tilde{L}\left(T(\tilde{x},\tilde{\bm{t}})\right)
= \sum_{n=1}^m\binom{m}{n}a^{m-n}
e^{ax}\frac{\partial L\left(x,\bm{t}\right)}{\partial t_n}e^{-ax}
\nonumber\\
&=\sum_{n=0}^m\binom{m}{n}a^{m-n}
\left[
e^{ax}\left(L^n\right)_{\geq 0}e^{-ax},\,e^{ax}Le^{-ax}
\right]
=\left[\left(\tilde{L}^m\right)_{\geq 0},\,\tilde{L}\right],
\end{align}
which proves the desired result.
\end{proof}
\begin{cor}[Galilean symmetry of the KP equation]
Let $u(x,y,t)$ be a solution of the KP equation \eqref{KPeq}. 
Then $\tilde{u}(x,y,t)=u(x+2a y+3a^2 t,y+3a t,t)$ is also a solution 
with arbitrary constant $a$. 
\end{cor}

\subsection{Darboux transformation for the KP hierarchy}
Let $L\in\mathcal{E}^{(1)}$ of the form \eqref{def:generalizedLaxOp}
be a solution of the KP hierarchy, and $\psi$ an associated wave function 
of the form \eqref{def:BAfcn}. 
In this subsection, we consider a kind of gauge transformation
\begin{equation}
L \mapsto \tilde{L}=\mathrm{Ad}_G(L)=GLG^{-1}, \quad \psi \mapsto \tilde{\psi} = G\psi, 
\label{def:DarbouxTrf}
\end{equation}
where $G\in\mathcal{E}$ is an invertible pseudo-differential operator. 
If $\tilde{L}$ is again another solution of the KP hierarchy and 
$\tilde{\psi}$ is the wave function associated to $\tilde{L}$, 
we call the transformation \eqref{def:DarbouxTrf} a \textit{Darbourx transformation}.
Darboux transformations for the KP hierarchy have been discussed by 
many authors \cite{ChauShawYen,Helminck_vdLeur_CMP,Helminck_vdLeur_CJM,Helminck_vdLeur_PublRIMS,
LeviRagnisco,Oevel,OevelRogers,OevelSchief}. 
Among those, hereafter we shall follow the pseudo-differential operator formulation 
by Helminck and van de Leur 
\cite{Helminck_vdLeur_CMP,Helminck_vdLeur_CJM,Helminck_vdLeur_PublRIMS} (cf. \cite{OevelSchief}).

\begin{thm}[Elementary Darboux transformation for the KP hierarchy 
\cite{Helminck_vdLeur_CMP,Helminck_vdLeur_CJM,Helminck_vdLeur_PublRIMS,OevelSchief}]
\label{thm:HvdL_Prop4.1}
Let $W\in\mathcal{E}^{(0)}$ be a monic pseudo-differential operator
that satisfies the Sato-Wilson equation \eqref{SatoWilsonEqn}. 
Let $\varphi(x,\bm{t})$ satisfy
\begin{equation}
\frac{\partial \varphi(x,\bm{t})}{\partial t_n} = B_n \varphi(x,\bm{t}), \quad
B_n = \left(W\partial_x^nW^{-1}\right)_{\geq 0}.
\label{def:eigenfunction:W}
\end{equation}
We say $\varphi(x,\bm{t})$ is an eigenfunction of the Lax operator $L=W\partial_xW^{-1}$ 
if it satisfies \eqref{def:eigenfunction:W}. 

Define a first order differential operator $G_\varphi$ as
\begin{equation}
G_\varphi = \varphi\partial_x\varphi^{-1} = \partial_x-\frac{\varphi'}{\varphi}
=\partial_x - \left(\log\varphi\right)', 
\end{equation}
which satisfies
\begin{align}
G_\varphi\left(\varphi\right) &= 0,\\
\partial_n\left(G_\varphi\right) &= -\partial_n\partial_x\left(\log\varphi\right)
=-\left[\partial_x,\partial_n\left(\log\varphi\right)\right]
=-\left[G_\varphi,\partial_n\left(\log\varphi\right)\right].
\label{dn(Gphi)=...}
\end{align}
Using $G_\varphi$, we define another monic pseudo-differential operator 
$W_\varphi\in\mathcal{E}^{(0)}$ as 
\begin{equation}
W_\varphi = G_\varphi W\partial_x^{-1}.
\end{equation}
Then $W_\varphi$ also solves the Sato-Wilson equation \eqref{SatoWilsonEqn}. 
\end{thm}
A proof of Theorem \ref{thm:HvdL_Prop4.1} can be found in 
\cite{Helminck_vdLeur_CMP} (Proposition 4.1).

\begin{thm}[Darboux transformation in terms of $\tau$-functions
\cite{Helminck_vdLeur_CMP,Helminck_vdLeur_CJM,Helminck_vdLeur_PublRIMS}]
\label{thm:HvdL_BDT(tau)}
Denote by $\tau(x,\bm{t})$ (resp. $\tau_\varphi(x,\bm{t})$)
the $\tau$-function corresponds to $W$ (resp. $W_\varphi$). 
Then $\tau_\varphi(x,\bm{t})$ can be written as 
\begin{equation}
\tau_\varphi(x,\bm{t}) = \varphi(x,\bm{t})\tau(x,\bm{t}).
\end{equation}
\end{thm}
A proof of Theorem \ref{thm:HvdL_BDT(tau)} can be found in 
\cite{Helminck_vdLeur_CJM,Helminck_vdLeur_PublRIMS}.

\begin{lemma}[Darboux transformation for eigenfunctions]
\label{lemma:DT_eigenFcn}
Let $W\in\mathcal{E}^{(0)}$ be a monic pseudo-differential operator
that satisfies the Sato-Wilson equation \eqref{SatoWilsonEqn}, and let 
$\varphi_1(x,\bm{t})$, $\varphi_2(x,\bm{t})$ be eigenfunctions of
$L=W\partial_xW^{-1}$. 
The corresponding differential operators are abbreviated as 
$G_i=G_{\varphi_i}=\varphi_i\partial_x\varphi_i^{-1}$ ($i=1,2$).
Define $\varphi_2^{[1]}$ as
\begin{equation}
\varphi_2^{[1]} = G_1\varphi_2 =
\left(\partial_x-\frac{\varphi'_1}{\varphi_1}\right)\varphi_2.
\end{equation}
Then $\varphi_2^{[1]}$ is an eigenfunction of 
$L_{\varphi_1}=W_{\varphi_1}\partial_xW_{\varphi_1}^{-1}$.
\end{lemma}
\begin{proof} 
For simplicity, we use the notation 
$\partial_n(\varphi)=\frac{\partial\varphi}{\partial t_n}$ ($n=1,2,\ldots$). 
By the assumption, $\varphi_i$ ($i=1,2$) satisfies 
\begin{equation}
\partial_n(\varphi_i) = \left(W\partial_x^nW\right)_{\geq 0}\left(\varphi_i\right)
= \mathrm{Res}_{\partial_x}\left(W\partial_x^nW\varphi_i\partial_x^{-1}\right),
\end{equation}
where we have used \eqref{lemma6.1:(h)}. 

On the other hand, straightforward calculation with Lemma \ref{lemma:Helminck_vdLeur} shows
\begin{align}
&\left(W_{\varphi_1}\partial_x^nW_{\varphi_1}^{-1}\right)_{<0} 
=\left(\varphi_1\partial_x\varphi_1^{-1}W\partial_x^n
W^{-1}\varphi_1\partial_x^{-1}\varphi_1^{-1}\right)_{<0}
\nonumber\\
&=\varphi_1\partial_x\varphi_1^{-1}\left(W\partial_x^n
W^{-1}\varphi_1\partial_x^{-1}\right)_{<0}\varphi_1^{-1}
-\varphi_1^{-1}\mathrm{Res}_{\partial_x}\left(W\partial_x^n
W^{-1}\varphi_1\partial_x^{-1}\right)
\nonumber\\
&=\varphi_1\partial_x\varphi_1^{-1}
\left(W\partial_x^nW^{-1}\right)_{<0}\varphi_1\partial_x^{-1}\varphi_1^{-1}
\nonumber\\
&\qquad +\varphi_1\partial_x\varphi_1^{-1}
\left\{\left(W\partial_x^nW^{-1}\right)_{\geq 0}\left(\varphi_1\right)\right\}
\partial_x^{-1}\varphi_1^{-1}
-\varphi_1^{-1}\left\{\left(W\partial_x^nW^{-1}\right)_{\geq 0}\left(\varphi_1\right)\right\}
\nonumber\\
&=G_1\left(W\partial_x^nW^{-1}\right)_{<0}G_1^{-1}
+G_1
\left\{\partial_n\left(\log\varphi_1\right)\right\}
G_1^{-1}
-\partial_n\left(\log\varphi_1\right).
\end{align}
Thus we have
\begin{equation}
\left(W_{\varphi_1}\partial_x^nW_{\varphi_1}^{-1}\right)_{\geq 0} 
=G_1\left(W\partial_x^nW^{-1}\right)_{\geq 0}G_1^{-1}
-G_1\left\{\partial_n\left(\log\varphi_1\right)\right\}G_1^{-1}
+\partial_n\left(\log\varphi_1\right), 
\end{equation}
and obtain
\begin{align}
&\left(W_{\varphi_1}\partial_x^nW_{\varphi_1}^{-1}\right)_{\geq 0} \varphi_2^{[1]}
=\left(W_{\varphi_1}\partial_x^nW_{\varphi_1}^{-1}\right)_{\geq 0}G_1\varphi_{2}
\nonumber\\
&=G_1\left(W\partial_x^nW^{-1}\right)_{\geq 0}\varphi_2
-\left[G_1,\,\partial_n\left(\log\varphi_1\right)\right]\varphi_2
\nonumber\\
&=G_1\partial_n\left(\varphi_2\right)+\left(\partial_n G_1\right)\varphi_2
=\partial_n\left(G_1\varphi_2\right) 
=\partial_n\left(\varphi_2^{[1]}\right),
\end{align}
which completes the proof.
\end{proof}

{}From Lemma \ref{lemma:DT_eigenFcn}, we know that $G_2^{[1]}=G_{\varphi_2^{[1]}}$ can be applied to 
$L_{\varphi_1}=G_1LG_1^{-1}$ and obtain yet another solution
$G_2^{[1]}G_1LG_1^{-1}G_2^{[1]-1}$. 
This can be generalized to $N$-fold iteration of the
elementary Darboux transformations of Theorem \ref{thm:HvdL_Prop4.1}. 
To describe $N$-fold iteration, we introduce the following notation:
\begin{equation}
\begin{aligned}
\varphi^{[1]}_i &:= G_1\varphi_i, & G^{[1]}_i &:= G_{\varphi^{[1]}_i},
\\
\varphi^{[21]}_i &:= G^{[1]}_2\varphi^{[1]}_i= G^{[1]}_2G_1\varphi_i, 
& G^{[21]}_i &:= G_{\varphi^{[21]}_i},
\\
\varphi^{[321]}_i &:= G^{[21]}_3\varphi^{[21]}_i= G^{[21]}_3G^{[1]}_2G_1\varphi_i, 
& G^{[321]}_i &:= G_{\varphi^{[321]}_i},
\\
\vdots && \vdots\\
\varphi^{[N,\ldots,1]}_i &:= G^{[N-1,\cdots, 1]}_N \cdots G^{[1]}_2G_1\varphi_i, 
& G^{[N,\ldots,1]}_i &:= G_{\varphi^{[N,\ldots,1]}_i},\\
\vdots && \vdots
\end{aligned}
\end{equation}
Iterative actions of the elementary Darboux transformations may be summarized
in the following diagram:
\begin{equation}
\label{DTdiagram}
\begin{CD}
L @>{G_1}>> \mathrm{Ad}_{G_1}(L)
@>{G^{[1]}_2}>> \mathrm{Ad}_{G^{[1]}_2G_1}(L)
@>{G^{[21]}_3}>> \mathrm{Ad}_{G^{[21]}_3G^{[1]}_2G_1}(L)
@>{G^{[321]}_4}>> \cdots
\\
\varphi_1 @>{G_1}>> 0
\\
\varphi_2 @>{G_1}>> \varphi^{[1]}_2
@>{G^{[1]}_2}>> 0
\\
\varphi_3 @>{G_1}>> \varphi^{[1]}_3
@>{G^{[1]}_2}>> \varphi_3^{[21]}
@>{G^{[21]}_3}>> 0
\\
\varphi_4 @>{G_1}>> \varphi^{[1]}_4
@>{G^{[1]}_2}>> \varphi_4^{[21]}
@>{G^{[21]}_3}>> \varphi_4^{[321]}
@>{G^{[321]}_4}>> 0 \; .
\\
\vdots @. \vdots @. \vdots @. \vdots 
\end{CD}
\end{equation}

The following proposition is a generalization of Crum's theorem \cite{Crum}
(cf. \cite{MatveevSalle}) to the KP hierarchy.
\begin{thm}[Iteration of the elementary Darboux transformations \cite{LeviRagnisco,Oevel,OevelSchief}]
\label{thm:Crum'sThm}
Let $W\in\mathcal{E}^{(0)}$ be a monic pseudo-differential operator
that satisfies the Sato-Wilson equation \eqref{SatoWilsonEqn}. 
Let $\varphi_1(x,\bm{t}),\ldots,\varphi_N(x,\bm{t})$ are eigenfunction 
of the Lax operator $L=W\partial_xW^{-1}$.
Denote by $G_N$ the monic differential operator of order $N$ that 
corresponds to 
the iteration of the elementary Darboux transformations 
associated with $\varphi_1(x,\bm{t}),\ldots,\varphi_N(x,\bm{t})$:
\begin{equation}
G_N = G^{[N-1\ldots 1]}_N\cdots G^{[1]}_2G_1
= \partial_x^N + \tilde{g}_{1} \partial_x^{N-1} + \cdots + \tilde{g}_N.
\end{equation}
Then the coefficients $\tilde{g}_1,\ldots,\tilde{g}_N$ are parameterized as 
\begin{equation}
\tilde{g}_{N-i} = (-1)^{N-i}
\frac{\mathcal{W}_i\left(\varphi_1,\ldots,\varphi_N\right)}{
\mathcal{W}\left(\varphi_1,\ldots,\varphi_N\right)}\quad (i=0,1,\ldots,N-1), 
\label{gi=det/Wr}
\end{equation}
where $\mathcal{W}$ is the usual Wronskian determinant and 
\begin{equation}
\mathcal{W}_i\left(\varphi_1,\ldots,\varphi_N\right)
=\begin{bmatrix}
\varphi_1 & \cdots & \partial_x^{i-1}\varphi_1& \partial_x^{i+1}\varphi_1
& \cdots & \partial_x^{N}\varphi_1\\
\varphi_2 & \cdots & \partial_x^{i-1}\varphi_2& \partial_x^{i+1}\varphi_2
& \cdots & \partial_x^{N}\varphi_2\\
\vdots & & \vdots & \vdots & & \vdots\\
\varphi_N & \cdots & \partial_x^{i-1}\varphi_N& \partial_x^{i+1}\varphi_N
& \cdots & \partial_x^{N}\varphi_N
\end{bmatrix}.
\end{equation}
\end{thm}
\begin{proof}
{}From the diagram \eqref{DTdiagram}, we know that 
\begin{equation}
G_N\varphi_i=0 \quad (i=1,2,\ldots,N). 
\label{GNvarphii=0}
\end{equation}
Solving the linear algebraic equation \eqref{GNvarphii=0} by the Cramer's rule, 
we obtain \eqref{gi=det/Wr}. 
\end{proof}

Transformed wave operator can be calculated more explicitly:
\begin{equation}
\label{transformedW}
\begin{aligned}
G_N W\partial_x^{-N} &=
\left(\partial_x^N + \tilde{g}_1\partial_x^{N-1}+\cdots + \tilde{g}_N\right)
\left(1+w_1\partial_x^{-1}+w_2\partial_x^{-2}+\cdots\right)\partial_x^{-N}
\\
&=1+\left(w_1+\tilde{g}_1\right)\partial_x^{-1}
+\left(Nw_1'+w_2+\tilde{g}_1w_1+\tilde{g}_2\right)\partial_x^{-2}
+\cdots.
\end{aligned}
\end{equation}
The solution of the KP equation that corresponds to 
\eqref{transformedW} is of the form
\begin{equation}
u = -\left(w_1+\tilde{g}_1\right)'
= -w_1' +\partial_x^2\log\mathcal{W}\left(\varphi_1,\ldots,\varphi_N\right).
\label{N-foldtransformedSolofKPeq}
\end{equation}

We can describe the action of the transformation $G_N$ to the $\tau$-function.
\begin{cor}
\label{cor:transformedTauFcn}
Let $L$ be a solution to the KP hierarchy and $\tau$ the $\tau$-function 
that corresponds to $L$. 
Denote by $\tilde{\tau}$ the $\tau$-function that corresponds to
$\mathrm{Ad}_{G_N}\left(L\right)$. The transformed $\tau$-function 
$\tilde{\tau}$ is given by
\begin{equation}
\tilde{\tau} = \mathcal{W}\left(\varphi_1,\ldots,\varphi_N\right)\tau.
\label{tildetau=Wtau}
\end{equation}
\end{cor}
 \begin{proof}
{}From Theorem \ref{thm:Crum'sThm}, we have
\begin{align}
\varphi^{[N-1,\ldots 1]}_N 
&= G_{N-1}\varphi_N
=\partial_x^{N-1}\varphi_N + \sum_{i=0}^{N-2}
(-1)^{N-1-i}
\frac{\mathcal{W}_i\left(\varphi_1,\ldots,\varphi_{N-1}\right)}{
\mathcal{W}\left(\varphi_1,\ldots,\varphi_{N-1}\right)}\partial_x^{i}\varphi_N
\nonumber\\
&= \frac{\mathcal{W}\left(\varphi_1,\ldots,\varphi_{N-1},\varphi_N\right)}{
\mathcal{W}\left(\varphi_1,\ldots,\varphi_{N-1}\right)}, 
\end{align}
where we have applied the expansion along the $N$-th row on 
$\mathcal{W}\left(\varphi_1,\ldots,\varphi_{N-1},\varphi_N\right)$. 
Applying Theorem \ref{thm:HvdL_BDT(tau)} recursively, we obtain
\begin{align}
\tilde{\tau}
&= \varphi^{[N-1\ldots 1]}_N\cdots \varphi^{[1]}_2\varphi_1\tau
\nonumber\\
&=\frac{\mathcal{W}\left(\varphi_1,\ldots,\varphi_{N-1},\varphi_N\right)}{
\mathcal{W}\left(\varphi_1,\ldots,\varphi_{N-1}\right)} 
\cdots
\frac{\mathcal{W}\left(\varphi_1,\varphi_2,\varphi_3\right)}{
\mathcal{W}\left(\varphi_1,\varphi_2\right)} \cdot
 \frac{\mathcal{W}\left(\varphi_1,\varphi_2\right)}{\varphi_1}\cdot\varphi_1\tau,
\end{align}
which is equivalent to \eqref{tildetau=Wtau}. 
\end{proof}

\subsection{Reduction of the KP hierarchy}
Let $L\in\mathcal{E}^{(1)}$ be a solution to the KP hierarchy.
For a positive integer $\ell$, we say $L$ is $\ell$-reduced 
if $L$ satisfies 
\begin{equation}
\left(L^\ell\right)_{<0} = 0.
\label{l-reduction}
\end{equation}
Under the condition \eqref{l-reduction}, it is clear that
\begin{equation}
\left(L^{\ell}\right)_{<0} = \left(L^{2\ell}\right)_{<0} = 
\left(L^{3\ell}\right)_{<0} = \cdots =0
\end{equation}
and hence we have
\begin{equation}
\frac{\partial L}{\partial t_{\ell}} 
=\frac{\partial L}{\partial t_{2\ell}} 
=\frac{\partial L}{\partial t_{3\ell}} =\cdots=0.
\label{l-reduced:dL/dtml=0}
\end{equation}
In the case $\ell=2$, the $2$-reduced hierarchy is equivalent to the KdV hierarchy. 

We consider the elementary Darboux transformation 
(Theorem \ref{thm:HvdL_Prop4.1}) for the $\ell$-reduced case.
\begin{lemma}
\label{lemma:EDT_reducedCase}
Fix a positive integer $m$ and 
let $L$ be a solution of the KP hierarchy that satisfies the condition
\begin{equation}
\frac{\partial L}{\partial t_{m}} =0.
\label{dL/dtm=0}
\end{equation}
Let $\varphi(x,\bm{t})$ be an eigenfunction of $L$, i.e., 
\begin{equation}
\frac{\partial\varphi(x,\bm{t})}{\partial t_n} = \left(L^n\right)_{\geq 0}\varphi(x,\bm{t})
\quad (n=2,3,\ldots).
\end{equation}
We further assume $\varphi(x,\bm{t})$ satisfies
\begin{equation}
\frac{\partial\varphi(x,\bm{t})}{\partial t_{m}} 
=\mu_{m}\varphi(x,\bm{t})
\label{reducedEigenFcn}
\end{equation}
where $\mu_{m}$ is independent of $x$. 
Then 
$\tilde{L} = G_\varphi LG_\varphi^{-1}$
also satisfies 
\begin{equation}
\frac{\partial\tilde{L}}{\partial t_{m}} =0.
\label{l-reduced:d(tildeL)/dtm=0}
\end{equation}
\end{lemma}
\begin{proof}
{}From \eqref{dn(Gphi)=...} and \eqref{reducedEigenFcn}, it follows that
\begin{equation}
\partial_{m}\left(G_\varphi\right) 
=-\left[G_\varphi,\partial_{m}\left(\log\varphi\right)\right]
=-\left[G_\varphi,\mu_{m}\right] = 0.
\end{equation}
Together with \eqref{dL/dtm=0},
we have \eqref{l-reduced:d(tildeL)/dtm=0}.
\end{proof}
Lemma \ref{lemma:EDT_reducedCase} shows that under the condition
\eqref{reducedEigenFcn} for $m=\ell,2\ell,3\ell,\ldots$, the elementary Darboux transformation
preserves the differential equations for the $\ell$-reduced KP hierarchy, 
\eqref{generalizedLaxEqs} and \eqref{l-reduced:dL/dtml=0}.
We will use this fact to discuss a class of elliptic soliton solutions 
for the KdV hierarchy.

\section{Stationary solution of the KP hierarchy associated with the Lam\'e function}
In this section, 
we consider a special solution to the KP hierarchy that is a generalization of 
the stationary solution \eqref{stationarySol_u(x,t)} to the whole hierarchy. 
Firstly we recall the addition formulas for the Weierstrass functions 
\cite{PastrasBook,AbramowitzStegun,WhittakerWatson}:
\begin{align}
\wp(x\pm k) &= \frac{1}{4}\left\{\frac{\wp'(x)\mp\wp'(k)}{\wp(x)-\wp(k)}\right\}^2
-\wp(x)-\wp(k),
\label{additionFormula:wp}\\
\zeta(x\pm k) &= \zeta(x)\pm\zeta(k)+\frac{1}{2}\frac{\wp'(x)\mp\wp'(k)}{\wp(x)-\wp(k)}.
\label{additionFormula:zeta}
\end{align}
Straightforward calculation with 
\eqref{additionFormula:wp} and \eqref{additionFormula:zeta}
shows \eqref{def:LameFcn} satisfies
\begin{align}
\Phi'(x;k) &= \frac{1}{2}\frac{\wp'(x)-\wp'(k)}{\wp(x)-\wp(k)}\Phi(x;k),
\label{Phi'}\\
\Phi''(x;k)&= \left\{2\wp(x)+\wp(k)\right\}\Phi(x;k).
\label{Phi''}
\end{align} 
It should be remarked that the second equation \eqref{Phi''} is 
a special case of the Lam\'e equation in Weierstrass form:
\begin{equation}
y'' = \left\{A+B\wp(x)\right\}y = 0.
\end{equation}

To calculate 
higher order derivatives of the Lam\'e function \eqref{def:LameFcn}, 
we prepare several notations. 
Let $R=\mathbb{C}\left[\wp(x),\wp'(x)\right]$ be 
the ring of polynomials with respect to 
$\wp(x)$, $\wp'(x)$, and 
$\mathcal{D}_R = R\left[\partial_x\right]$ the ring of differential operators
with coefficients in $R$. 
For $P\left(x,\partial_x\right)=\sum_{j=0}^nf_j(x)\partial_x^j\in\mathcal{D}_R$, 
we use the notation $P'\left(x,\partial_x\right)=\sum_{j=0}^nf'_j(x)\partial_x^j$.
Hereafter we will omit $x$ and $\partial_x$ of $P\left(x,\partial_x\right)$ 
if it does not cause confusion.

Define $P_2\left(x,\partial_x\right)$, 
$Q_1\left(x,\partial_x\right)\in\mathcal{D}_R$ as
\begin{equation}
P_2 = \partial_x^2-2\wp(x), \quad Q_1 = \wp(x)\partial_x -\frac{1}{2}\wp'(x).
\label{P2Q1}
\end{equation}
It follows from \eqref{Phi'} and \eqref{Phi''} that
\begin{equation}
P_2\Phi(x;k) = \wp(k)\Phi(x;k), \quad 
\wp(k)\Phi'(x;k) = \left\{Q_1+\frac{1}{2}\wp'(k)\right\}\Phi(x;k).
\label{P2Phi=...,wp(k)Phi'=...}
\end{equation}
We now generalize this observation.
\begin{lemma}
\label{thm:PnPhi=alphaPhi}
For $n=2,3,\ldots$, 
the Lam\'e function \eqref{def:LameFcn} satisfies the relations
\begin{align}
P_n\Phi(x;k) &= \alpha_n(k)\Phi(x;k),\label{PnPhi=...}\\
\alpha_n(k)\Phi'(x;k) &= \left\{Q_{n-1}+\alpha_{n+1}(k)\right\}\Phi(x;k),
\label{alphanPhi'=...}
\end{align}
where $\alpha_2(k)$, $\alpha_3(k),\ldots$ are defined by 
\begin{equation}
\alpha_{2m}(k) = \wp(k)^m, \quad 
\alpha_{2m+1}(k) = \frac{1}{2}\wp(k)^{m-1}\wp'(k) \quad (m=1,2,\ldots),
\label{def:alpha}
\end{equation}
and $P_n,Q_{n-1}\in\mathcal{D}_R$ are defined by the following recursion relations
with the initial condition \eqref{P2Q1}:
\begin{equation}
P_{n+1} = P_{n}\partial_x + P'_{n}-Q_{n-1}, \quad
Q_n = 2\wp(x)P_{n}-Q_{n-1}\partial_x-Q'_{n-1}
\quad (n=2,3,\ldots).
\label{recursion:PnQn}
\end{equation}
\end{lemma}

The assertions of Lemma \ref{thm:PnPhi=alphaPhi} can be proved by induction 
with differentiating \eqref{PnPhi=...}, \eqref{alphanPhi'=...} with respect to $x$, 
so we omit the details. 
Examples of $P_n$, $Q_{n-1}$ ($n=2,3,\ldots$) are listed in Appendix. 
We remark that $P_n$, $Q_{n-1}$ ($n=2,3,\ldots$) are independent of $k$.

\begin{lemma}
\label{lemma:Rn(x)}
Define a sequences of polynomials $\left\{R_n(x)\right\}$ ($n=0,1,\ldots$) 
by the following recursion relation: 
\begin{align}
R_{n+1}(x) &= \left(6x^2-\frac{g_2}{2}\right)R'_n(x)+\left(4x^3-g_2x-g_3\right)R''_n(x)
\label{recurrence:Rn}
\end{align}
with the initial condition $R_0(x)=x$. Then we have
\begin{equation}
\frac{d^{2n}\wp(x)}{dx^{2n}} = R_n\!\left(\wp(x)\right), \quad
\frac{d^{2n+1}\wp(x)}{dx^{2n+1}} =\wp'(x) R'_n\!\left(\wp(x)\right)
\quad (n=0,1,\ldots).
\label{formula:wp{(n)}}
\end{equation}
\end{lemma}
\begin{lemma}
\label{lemma:Rn(x)=(2n-1)!x{n+1}}
The degree of $R_n(x)$ is $n+1$ and the coefficient of 
$x^{n+1}$, $x^n$ are $(2n+1)!$, 0, respectively.
\end{lemma}

We omit the proof for 
Lemma \ref{lemma:Rn(x)} and Lemma \ref{lemma:Rn(x)=(2n-1)!x{n+1}}
since it is straightforward by induction. 
Examples of $R_n$ ($n=0,1,2,\ldots$) are also listed in Appendix. 
{}From Lemma \ref{lemma:Rn(x)=(2n-1)!x{n+1}}, one can write $R_n(x)$ as
\begin{equation}
R_n(x) = (2n+1)!\left(x^{n+1} + \sum_{j=1}^nr_{n,j}x^{n-j}\right).
\label{Rn(x)=...}
\end{equation}
Examples of the formula \eqref{Rn(x)=...} are listed in Appendix. 

\begin{thm}
Define $\tilde{P}_n\left(x,\partial_x\right)\in\mathcal{D}_R$ ($n=0,1,2,\ldots$) as 
\begin{align}
\tilde{P}_{2n+2}\left(x,\partial_x\right)
&= P_{2n+2}\left(x,\partial_x\right) 
+ \sum_{j=1}^{n-1}r_{n,j}P_{2j}\left(x,\partial_x\right)+r_{n,0},\\
\tilde{P}_{2n+3}\left(x,\partial_x\right) &=
P_{2n+3}\left(x,\partial_x\right) 
+ \sum_{j=1}^{n-1} \frac{j\,r_{n,j}}{n+1}P_{2j+1}\left(x,\partial_x\right).
\end{align}
Then these operators satisfy
\begin{equation}
\tilde{P}_{n+2}\left(x,\partial_x\right)\Phi(x;k) = \frac{\wp^{(n)}(k)}{(n+1)!}\Phi(x;k)
\quad (n=0,1,2,,\ldots).
\label{tildeP(n+1)Phi=...}
\end{equation}
\end{thm}
\begin{proof}
Applying \eqref{Rn(x)=...} to \eqref{formula:wp{(n)}}, we obtain
\begin{align}
\frac{1}{(2n+1)!}\frac{d^{2n}\wp(x)}{dx^{2n}} &= \left\{
\wp(x)^{n+1} + \sum_{j=0}^{n-1}r_{n,j}\wp(x)^{j}\right\},\\
\frac{1}{(2n+1)!}\frac{d^{2n+1}\wp(x)}{dx^{2n+1}} &=\wp'(x) \left\{
(n+1)\,\wp(x)^{n} + \sum_{j=1}^{n-1} j\,r_{n,j}\wp(x)^{j-1}
\right\}
\end{align}
for $n=0,1,2,\ldots$
Substituting $x=k$ and applying \eqref{def:alpha}, we have
\begin{equation}
\begin{aligned}
\frac{\wp^{(2n)}(k)}{(2n+1)!}&= \alpha_{2n+2}(k) 
+ \sum_{j=1}^{n-1}r_{n,j}\alpha_{2j}(k)+r_{n,0},\\
\frac{\wp^{(2n+1)}(k)}{(2n+2)!} &= \alpha_{2n+3}(k) 
+ \sum_{j=1}^{n-1} \frac{j\,r_{n,j}}{n+1}\alpha_{2j+1}(k).
\end{aligned}
\label{recuesion:dn_wp}
\end{equation}
The relation \eqref{tildeP(n+1)Phi=...} is a direct consequence of 
\eqref{PnPhi=...} and \eqref{recuesion:dn_wp}.
\end{proof}

We recall the Laurent expansion of $\wp(x)$ and $\zeta(x)$ around $x=0$:
\begin{equation}
\wp(x)= \frac{1}{x^2} + \sum_{n=0}^\infty a_{n}x^{n}, \quad
\zeta(x) = \frac{1}{x} -\sum_{n=1}^\infty \frac{a_{n}}{n+1}x^{n+1}, 
\end{equation}
where the coefficients $a_0$, $a_1$, $a_2,\ldots$ are determined by 
the following relations:
\begin{equation}
\begin{aligned}
a_0 &= 0, \quad a_1= a_3 = a_5 = \cdots = a_{2n-1} = \cdots = 0,\\
a_2 &= \frac{g_2}{20}, \quad 
a_4 = \frac{g_3}{28}, \quad 
(n-1) (2 n+5)a_{2n+2} = 3\sum_{j=1}^{n-1} a_{2(n-j)}a_{2j}\quad (n=2,3,\ldots).
\end{aligned}
\label{recurrence:TaylorExp:wp}
\end{equation}
We remark that the recurrence relation in \eqref{recurrence:TaylorExp:wp} is 
a direct consequence of \eqref{diffEq:WeierstrassP:2}. 
The Laurent expansion of 
$\wp^{(n)}(x)=\frac{d^n\wp(x)}{dx^n}$ ($n=0,1,2,\ldots$) around $x=0$ 
is of the form 
\begin{equation}
\label{dnwp(x)/dxn=...}
\frac{\wp^{(n)}(x)}{(n+1)!}= \frac{(-1)^n}{x^{n+2}} + \frac{a_n}{n+1}
+\sum_{m=1}^\infty\frac{(m+n)!}{m!(n+1)!}a_{m+n}x^m.
\end{equation}
If we set $k=-\lambda^{-1}$ and 
apply \eqref{dnwp(x)/dxn=...} to the right-hand-side of \eqref{tildeP(n+1)Phi=...}, 
we obtain
\begin{align}
\tilde{P}_{n+2}\left(x,\partial_x\right)\Phi\left(x;-\lambda^{-1}\right) = 
\left\{\lambda^{n+2} + \frac{a_n}{n+1}
+\sum_{m=1}^\infty\frac{(-1)^m(m+n)!}{m!(n+1)!}a_{m+n}\lambda^{-m}
\right\}\Phi\left(x;-\lambda^{-1}\right)
\label{tildeP(n+1)Phi=...withLaurentExpansion:lambda}
\end{align}
for $n=0,1,2,\ldots$. 

Now we will rewrite \eqref{tildeP(n+1)Phi=...withLaurentExpansion:lambda} as a
relation between pseudo-differential operators. Toward this aim, 
we firstly rewrite $\Phi(x;k)$ as follows:
\begin{align}
\Phi\left(x;-\lambda^{-1}\right) &= \hat{v}_0(x;\lambda)\exp\left(\lambda x\right),
\label{Phi(x;-lambda(-1))=v0exp(-x/k)}\\
\hat{v}_0(x;\lambda) &=
\frac{\sigma(x-\lambda^{-1})}{\sigma(x)}\exp\left[\left(\zeta(\lambda^{-1})-\lambda\right)x\right].
\end{align}
As a function of $\lambda$, $\hat{v}_0(x;\lambda)$ is regular at $\lambda=\infty$ 
(i.e. $k=0$) and satisfies 
$\lim_{\lambda\to\infty}\hat{v}_0(x;\lambda)=1$. The Laurent expansion of $\hat{v}_0(x;\lambda)$ 
around $\lambda=\infty$ is of the form
\begin{equation}
\hat{v}_0(x;\lambda) = 1 + \sum_{n=1}^\infty v_{0,n}(x)\lambda^{-n}, 
\label{expansion:hatv0:lambda}
\end{equation}
with 
\begin{equation}
\begin{aligned}
v_{0,1}(x) &= -\zeta(x), \quad v_{0,2}(x) = \frac{1}{2}\left\{\zeta(x)^2-\wp(x)\right\}, \\
v_{0,3}(x) &= \frac{1}{2}\wp(x)\zeta(x) + \frac{1}{6}\left\{\wp'(x)-\zeta(x)^3\right\}
-\frac{g_2}{60}x, \quad \ldots. 
\end{aligned}
\label{v01,v02,v03=...}
\end{equation}
We introduce a pseudo-differential operator $V_0(x,\partial_x)$ associated with 
\eqref{expansion:hatv0:lambda}: 
\begin{equation}
V_0(x,\partial_x) = 1 + \sum_{n=1}^\infty v_{0,n}(x)\partial_x^{-n}.
\end{equation}
Using $V_0(x,\partial_x)$, one can rewrite \eqref{Phi(x;-lambda(-1))=v0exp(-x/k)} as 
\begin{equation}
\Phi\left(x;-\lambda^{-1}\right) = V_0(x,\partial_x) e^{\lambda x}. 
\label{Phi=V0elambdax}
\end{equation}
We introduce another pseudo-differential operator $U_{n+2}(\partial_x)$ 
($n=0,1,2,\ldots$) as 
\begin{equation}
U_{n+2}(\partial_x) = 
\partial_x^{n+2} + \frac{a_n}{n+1}
+\sum_{m=1}^\infty\frac{(-1)^m(m+n)!}{m!(n+1)!}a_{m+n}\partial_x^{-m}. 
\label{def:U(n+2)=...}
\end{equation}
Since the coefficients of $U_{n+2}(\partial_x)$ are constants, 
it commutes with $\partial_x$. 
Thus it follows from 
\eqref{tildeP(n+1)Phi=...withLaurentExpansion:lambda} 
and \eqref{Phi=V0elambdax} that
\begin{equation}
\tilde{P}_{n+2}\left(x,\partial_x\right)V_0\left(x,\partial_x\right)
=V_0\left(x,\partial_x\right)U_{n+2}(\partial_x).
\end{equation}
\begin{thm}
\label{thm:L0_is_a_stationary_sol}
Define $L_0\in\mathcal{E}^{(1)}$ as 
\begin{equation}
L_0 = V_0\left(x,\partial_x\right)\partial_xV_0\left(x,\partial_x\right)^{-1}.
\end{equation}
Then $L_0$ is a stationary solution of the KP hierarchy, i.e., it satisfies
\begin{equation}
\left[\left(L_0^{n+2}\right)_{\geq 0},\,L_0\right] = 0 \quad (n=0,1,2,\ldots).
\end{equation}
\end{thm}
\begin{proof}
{}From \eqref{def:U(n+2)=...}, we have
\begin{equation}
\begin{aligned}
\left(L_0^{n+2}\right)_{\geq 0} &= 
\left(V_0\partial_x^{n+2}V_0^{-1}\right)_{\geq 0}
=\left(V_0\left(U_{n+2}-\frac{a_n}{n+1}\right)V_0^{-1}\right)_{\geq 0}\\
&=\tilde{P}_{n+2}-\frac{a_n}{n+1}
=V_0U_{n+2}V_0^{-1}-\frac{a_n}{n+1}, 
\end{aligned}
\label{(L0(n+2))>0=tildeP(n+2)-an/(n+1)}
\end{equation}
and thus $\left(L_0^{n+2}\right)_{\geq 0}$ commutes with 
$L_0=V_0\partial_xV_0^{-1}$.
\end{proof}
We can calculate the explicit forms for the first few terms of $L_0$ and $L_0^2$ 
using \eqref{v01,v02,v03=...}:
\begin{align}
L_0 &= \partial_x -\wp(x)\partial_x^{-1}+\frac{1}{2}\wp'(x)\partial_x^{-2}
+\left(\frac{g_2}{10}-2\wp(x)^2\right)\partial_x^{-3}
+\cdots.
\label{L0=dx2+...}\\
L_0^2 &= \partial_x^2 -2\wp(x) -\frac{g_2}{20}\partial_x^{-2}
-6\wp(x)\wp'(x)\partial_x^{-3}+\cdots.
\end{align}
It is obvious that $\left(L_0^2\right)_{<0}\neq 0$ that means 
$L_0$ is not 2-reduced. However, as can be seen from Theorem 
\ref{thm:L0_is_a_stationary_sol}, $L_0$ satisfies the 
differential equations for the 2-reduced (KdV) hierarchy.

\section{Wave function and $\tau$-function associated with the stationary solution}
To construct non-trivial eigenfunctions of $L_0$, 
we introduce infinite time variables $\bm{t}=(t_2,t_3,\ldots)$ and  
define $\Psi(x,\bm{t};k)$ as 
\begin{equation}
\Psi(x,\bm{t};k) = \Phi(x;k)\exp\!\left[\sum_{j=0}^\infty
\left\{\frac{\wp^{(j)}(k)}{(j+1)!}-
\frac{a_j}{j+1}
\right\}t_{j+2}\right]. 
\label{def:PWF:Psi}
\end{equation}
\begin{lemma}
$\Psi(x,\bm{t};k)$ is an eigenfunction of $L_0$.
\end{lemma}
\begin{proof}
{}From the definition \eqref{def:PWF:Psi} and the relations 
\eqref{tildeP(n+1)Phi=...}, \eqref{(L0(n+2))>0=tildeP(n+2)-an/(n+1)}, 
we have 
\begin{align}
\frac{\partial\Psi(x,\bm{t};k)}{\partial t_{n+2}}
&=\left\{\frac{\wp^{(n)}(k)}{(n+1)!}-
\frac{a_{n}}{n+1}\right\}\Psi(x,\bm{t};k)
= \left\{\tilde{P}_{n+2}\left(x,\partial_x\right)-
\frac{a_{n}}{n+1}\right\}\Psi(x,\bm{t};k)
\nonumber\\
&=\left(L_0^{n+2}\right)_{\geq 0}\Psi(x,\bm{t};k)
\end{align}
for $n=0,1,2,\ldots$, and the assertion follows.
\end{proof}

\begin{thm}
$\Psi(x,\bm{t};k)$ solves the auxiliary linear problem 
\eqref{Lpsi=lambdapsi} and \eqref{dnpsi=Bnpsi} associated with $L_0$, 
where the auxiliary eigenvalue $\lambda$ is given by $\lambda=-k^{-1}$.
\end{thm}
\begin{proof}
We have already seen $\Psi(x,\bm{t};k)$ solves \eqref{dnpsi=Bnpsi}. 
To consider \eqref{Lpsi=lambdapsi}, we substitute $k=-\lambda^{-1}$ to \eqref{def:PWF:Psi} 
and apply \eqref{dnwp(x)/dxn=...}, \eqref{Phi(x;-lambda(-1))=v0exp(-x/k)}:
\begin{align}
\Psi\left(x,\bm{t};-\lambda^{-1}\right) &= \Phi\left(x;-\lambda^{-1}\right)
\exp\!\left[\sum_{j=0}^\infty
\left\{\frac{\wp^{(j)}(-\lambda^{-1})}{(j+1)!}-
\frac{a_j}{j+1}
\right\}t_{j+2}\right]
\nonumber\\
&=\hat{v}_0(x;\lambda)e^{\lambda x}
\exp\!\left[\sum_{j=0}^\infty\left\{\lambda^{j+2}
+\sum_{m=1}^\infty\frac{(-1)^j(m+j)!}{m!(j+1)!}a_{m+j}\lambda^{-m}\right\}t_{j+2}\right]
\nonumber\\
&=\hat{v}_0(x;\lambda)
\exp\!\left[\sum_{j=0}^\infty\omega_{j+2}(\lambda)t_{j+2}\right]
e^{\xi(x,\bm{t};\lambda)},
\end{align}
where we have defined $\omega_{j+2}(\lambda)$ by 
\begin{equation}
\omega_{j+2}(\lambda)=\sum_{m=1}^\infty\frac{(-1)^j(m+j)!}{m!(j+1)!}a_{m+j}\lambda^{-m}
\quad (j=0,1,2,\ldots).
\end{equation}

Since $\omega_n(\lambda)$ is regular at $\lambda=\infty$ 
and $\omega_n(\lambda=\infty)=0$, 
we have the following form of expansion:
\begin{equation}
\exp\!\left[\sum_{j=0}^\infty \omega_{j+2}(\lambda)t_{j+2}\right]
= 1 + \frac{c_1(\bm{t})}{\lambda} + \frac{c_2(\bm{t})}{\lambda^2} + \cdots.
\label{exp[sumomega(j+2)t(j+2)=...]}
\end{equation}
We introduce the pseudo-differential operator $C_n(\bm{t};\partial_x)$ as 
\begin{equation}
C_n(\bm{t};\partial_x) = 1+c_1(\bm{t})\partial_x^{-1}+c_2(\bm{t})\partial_x^{-2}+\cdots,
\end{equation}
which corresponds to \eqref{exp[sumomega(j+2)t(j+2)=...]}.
Using this operator, we can rewrite 
$\Psi(x,\bm{t};-\lambda^{-1})$ as 
\begin{equation}
\Psi(x,\bm{t};-\lambda^{-1}) = V_0(x;\partial_x)C_n(\bm{t};\partial_x)e^{\xi(x,\bm{t};\lambda)},
\label{Psi=V0(partialx)Cn(partialx)exp(xi)}
\end{equation}
because the coefficients $c_1(\bm{t})$, $c_2(\bm{t}),\ldots$ do not depend on $x$.
Thus we have
\begin{align}
L_0\Psi\left(x,\bm{t};-\lambda^{-1}\right) &= V_0\partial_x C_n e^{\xi}
= V_0C_n \partial_x e^{\xi}= \lambda V_0C_n e^{\xi} 
= \lambda\Psi\left(x,\bm{t};-\lambda^{-1}\right), 
\end{align}
which indicates \eqref{Lpsi=lambdapsi}.
\end{proof}

We now define $V(x,\bm{t};\partial_x)$ as 
\begin{equation}
V(x,\bm{t};\partial_x) = V_0(x;\partial_x)C_n(\bm{t};\partial_x), 
\label{V=V0Cn}
\end{equation}
which is the wave operator associated with the wave function 
$\Psi\left(x,\bm{t};-\lambda^{-1}\right)$. 
It follows that 
\begin{equation}
L_0 =V_0(x;\partial_x)\partial_xV_0(x;\partial_x)^{-1}
=V(x,\bm{t};\partial_x)\partial_xV(x,\bm{t};\partial_x)^{-1}, 
\label{VpartialxVinv=V0partialxV0inv}
\end{equation}
because $C_n(\bm{t};\partial_x)$ commutes with $\partial_x$. 
This means that both of $V(x,\bm{t};\partial_x)$ and $V_0(x;\partial_x)$ 
are wave operators that correspond to $L_0$.

Next we construct the $\tau$-function $\tau_V$ that corresponds to $V(x,\bm{t};\partial_x)$, 
i.e., 
\begin{align}
\frac{\tau_V\left(x-\lambda^{-1},\bm{t}-\left[\lambda^{-1}\right]\right)}{\tau_V(x,\bm{t})}
&=\hat{v}_0(x;\lambda)
\exp\!\left[\sum_{j=0}^\infty\omega_{j+2}(\lambda)t_{j+2}\right]
\nonumber\\
&=\frac{\sigma(x-\lambda^{-1})}{\sigma(x)}\exp\left[\omega_1(\lambda)x
+\sum_{j=0}^\infty\omega_{j+2}(\lambda)t_{j+2}\right],
\label{tauFcn_stationarySol}
\end{align}
where we have defined $\omega_1(\lambda)$ by
\begin{equation}
\omega_1(\lambda)=\zeta(\lambda^{-1})-\lambda
=\sum_{m=2}^\infty \frac{(-1)^ma_{m-1}}{m}\lambda^{-m}.
\end{equation}
To obtain explicit expression of $\tau_V$, 
we prepare a lemma. Here we use the convention $x=t_1$. 
\begin{lemma}
\label{lemma:Q(t)}
There exists a quadratic function $Q(\bm{t})$ of infinite variables of the form 
\begin{equation}
Q(x,\bm{t}) 
=\frac{1}{2}\sum_{i,j=1}^\infty q_{ij}t_it_j+ \sum_{i=2}^\infty s_i t_i, 
\quad x=t_1, \quad q_{ji}=q_{ij},
\end{equation}
which satisfies the following relation:
\begin{equation}
-Q\left(\bm{t}-\left[\lambda^{-1}\right]\right)+Q(x,\bm{t})
=\sum_{n=1}^\infty \omega_n(\lambda)t_n.
\label{-Q+Q=sum_ometantn}
\end{equation}
\end{lemma}
\begin{proof}
Denote by $\beta^{(m)}_n$ the coefficient of $\lambda^{-m}$ in $\omega_n(\lambda)$, i.e., 
\begin{equation}
\beta^{(m)}_n = 
\begin{cases}
\frac{(-1)^ma_{m-1}}{m} & (n=1), \\[3mm]
\frac{(-1)^m(m+n-2)! a_{m+n-2}}{m!(n-1)!} & (n=2,3,\ldots).
\end{cases}
\end{equation}
Straightforward calculation shows $m\beta_n^{(m)}=n\beta_m^{(n)}$ for 
$m,n=1,2,3,\ldots$. Thus we can set
\begin{equation}
q_{mn} = m\beta_n^{(m)}=n\beta_m^{(n)} = 
\begin{cases}
-a_{n-1} & (m=1,\,n=2,3,\ldots),\\
\frac{(-1)^m(m+n-2)! a_{m+n-2}}{(m-1)!(n-1)!} & (m,n=2,3,\ldots).
\end{cases}
\label{qmn}
\end{equation}
Under the conditions \eqref{qmn} and
\begin{equation}
s_n = \frac{n}{2}\sum_{j=1}^{n-1}
\frac{q_{n-j,j}}{j(n-j)} \quad (n=2,3,\ldots), 
\label{sn}
\end{equation}
we see that \eqref{-Q+Q=sum_ometantn} holds. 
\end{proof}
Lemma \eqref{lemma:Q(t)} implies the followin theorem. 
\begin{thm}
The $\tau$-function $\tau_V(x,\bm{t})$ of \eqref{tauFcn_stationarySol} is 
given by
\begin{equation}
\tau_V(x,\bm{t}) = e^{-Q\left(x,\bm{t}\right)}\sigma(x). 
\label{def:tauV}
\end{equation}
\end{thm}

Setting $t_n=0$ for $n\geq 6$ as an example,  we obtain
\begin{equation}
\label{Q(t1,...,t5)=...}
\begin{aligned}
Q&\left(x=t_1,t_2,\ldots,t_5\right)\\
&= a_2 t_2^2-3a_4 t_3^2+10a_6t_4^2-35a_8t_5^2
-a_2 x t_3-a_4 x t_5+4a_4t_2t_4-15a_6t_3t_5-\frac{a_2}{3}t_4
\\
&=\frac{g_2}{60}\left(3t_2^2-3x t_3-t_4\right)
+\frac{g_3}{28}\left(-3t_3^2-x t_5+4t_2t_4\right)
+\frac{g_2^2}{240}\left(2t_4^2-3t_3t_5\right)
-\frac{3 g_2 g_3}{176}t_5^2.
\end{aligned}
\end{equation}
If we substitute \eqref{def:tauV} with \eqref{Q(t1,...,t5)=...} 
to \eqref{u2=del2logtau,u3=...,u4=...}, 
the results coincide with \eqref{L0=dx2+...}.

We now apply Theorem \ref{thm:Crum'sThm} to $L_0=V_0\partial_xV_0^{-1}=V\partial_xV^{-1}$. 
We consider a set of $N$ eigenfunctions 
\begin{equation}
\varphi_i = 
 \sum_{j=1}^N \gamma_{ij}\Psi\left(x,\bm{t};k_j\right) \quad (i=1,2,\ldots,N), 
\end{equation}
where $k_i$ ($i=1,\ldots,N$), $\gamma_{ij}$ ($i,j=1,\ldots,N$) are constants.
Then one can construct a differential operator $G_N$ by the 
generalized Crum formula (Theorem \ref{thm:Crum'sThm}).
The corresponding solution \eqref{N-foldtransformedSolofKPeq} 
for the KP equations is given by
\begin{equation}
u = -v_1' +\partial_x^2\log\mathcal{W}\left(\varphi_1,\ldots,\varphi_N\right)
= -\wp(x) +\partial_x^2\log\mathcal{W}\left(\varphi_1,\ldots,\varphi_N\right),
\label{LiZhangSolution:KP:revisited}
\end{equation}
which is a generalization of \eqref{LiZhangSolution:KP}. 
Although the background of the solution \eqref{LiZhangSolution:KP:revisited} is stationary, 
solutions with the non-stationary background \eqref{GalileanTransformedWPsol}
can be constructed by applying 
Theorem \ref{thm:GalileanSymmetry} to \eqref{LiZhangSolution:KP:revisited}. 

The $\tau$-function $\tilde{\tau}_V$
that corresponds to \eqref{LiZhangSolution:KP:revisited}
can be derived by applying Corollary \ref{cor:transformedTauFcn} to 
$\tau_V$ of \eqref{def:tauV}:
\begin{equation}
\tilde{\tau}_V
= e^{-Q\left(x,t_2,\ldots\right)}\sigma(x)
\mathcal{W}\left(\varphi_1,\ldots,\varphi_N\right). 
\label{transformedTauFcn}
\end{equation}
Substituting $\tilde{f}=e^{Q}\tilde{\tau}_V$ with \eqref{Q(t1,...,t5)=...} in 
\eqref{bilinearKP_wp:2}--\eqref{bilinearKP_wp:deg6:2}, 
then the resulting equations for $\tilde{\tau}_V$ coincide with
those of the KP hierarchy listed in \cite{JimboMiwa}, i.e., 
the equations \eqref{bilinearKP_wp:2}--\eqref{bilinearKP_wp:deg6:2} 
with $g_2=g_3=0$. 

\section{KP-type hierarchy with elliptic background}
In the preceding section, we constructed a solution to the KP hierarchy
through the iterative application of Darboux transformations to the ``vacuum solution'' $V$.
Expressing this solution in terms of wave operators, it can be represented as 
\begin{equation}
W=G_NV\partial_x^{-N}, 
\label{W=GNVdx(-N)}
\end{equation}
which satisfies the Sato-Wilson equation \eqref{SatoWilsonEqn}.
Let us consider differential equation for $G_N = W\partial_x^N V^{-1}$. 
It is straightforward to show that
\begin{align*}
\frac{\partial G_N}{\partial t_n}&=B_n G_N - G_N B_n^{(V)}, \quad
B_n = \left(W\partial_x^nW^{-1}\right)_{\geq 0}
=\left(G_NL_0^nG_N^{-1}\right)_{\geq 0}.
\end{align*}

To generalize this observation, 
we consider the following factorized form:
\begin{equation}
W = \tilde{W}^{(1)}\tilde{W}^{(0)}\partial_x^{-m},
\label{W=tildeW1tildeW0dx(-m)}
\end{equation}
where $\tilde{W}^{(0)}\in\mathcal{E}^{(0)}$, 
$\tilde{W}^{(1)}\in\mathcal{E}^{(m)}$ ($m\in\mathbb{Z}_{\geq 0}$) and both are monic.
If we set $m=N$, $\tilde{W}^{(0)}=V$ and $\tilde{W}^{(1)}=G_N$, then 
\eqref{W=tildeW1tildeW0dx(-m)} coincides with \eqref{W=GNVdx(-N)}.

Assume $\tilde{W}^{(0)}$ satisfies 
the Sato-Wilson equation \eqref{SatoWilsonEqn}: 
\begin{equation}
\frac{\partial\tilde{W}^{(0)}}{\partial t_n} 
= \tilde{B}^{(0)}_n\tilde{W}^{(0)}-\tilde{W}^{(0)}\partial_x^n, 
\quad \tilde{B}^{(0)}_n = \left(\tilde{W}^{(0)}\partial_x(\tilde{W}^{(0)})^{-1}\right)_{\geq 0}, 
\label{SatoWilson-typeEq_W(0)}
\end{equation}
which is regarded as a background wave.
If $\tilde{W}^{(1)}$ satisfies 
\begin{equation}
\frac{\partial\tilde{W}^{(1)}}{\partial t_n} 
= \tilde{B}^{(1)}_n\tilde{W}^{(1)}-\tilde{W}^{(1)}\tilde{B}^{(0)}_n,
\quad \tilde{B}^{(1)}_n = \left(\tilde{W}^{(1)}\tilde{B}^{(0)}_n(\tilde{W}^{(1)})^{-1}\right)_{\geq 0}, 
\label{SatoWilson-typeEq_tildeB0n}
\end{equation}
then $W$ of \eqref{W=tildeW1tildeW0dx(-m)} satisfies 
the Sato-Wilson equation \eqref{SatoWilsonEqn} with 
\begin{equation}
B_n = \left(W\partial_x W^{-1}\right)_{\geq 0}
= \left(\tilde{W}^{(1)}(\tilde{L}^{(0)})^n (\tilde{W}^{(1)})^{-1}\right)_{\geq 0}, \quad
\tilde{L}^{(0)} = \tilde{W}^{(0)}\partial_x(\tilde{W}^{(0)})^{-1}.
\end{equation}
Conversely, if $W$ satisfies 
the Sato-Wilson equation \eqref{SatoWilsonEqn}, then 
$\tilde{W}^{(1)}=W\partial_x^m(\tilde{W}^{(0)})^{-1}$ satisfies 
\eqref{SatoWilson-typeEq_tildeB0n}.

As we have mentioned, \eqref{W=GNVdx(-N)} gives an example with $m=N$. 
Next we consider the case $m=0$ and define
\begin{equation}
\tilde{W}^{(j)} = 1 + \tilde{w}^{(j)}_1\partial_x^{-1}+ \tilde{w}^{(j)}_2\partial_x^{-2}+\cdots
\quad (j=0,1).
\end{equation}
In this case, $\tilde{B}^{(0)}_2$, $\tilde{B}^{(0)}_3$, 
$\tilde{B}^{(1)}_2$, $\tilde{B}^{(1)}_3$ are given by
\begin{equation}
\begin{aligned}
\tilde{B}^{(0)}_2 &= \partial_x^2 -2 (\tilde{w}^{(0)}_1)',\\
\tilde{B}^{(0)}_3 &= \partial_x^2 -3 (\tilde{w}^{(0)}_1)'\partial_x
-3(\tilde{w}^{(0)}_2)'+3\tilde{w}^{(0)}_1(\tilde{w}^{(0)}_1)'-3(\tilde{w}^{(0)}_1)'',\\
\tilde{B}^{(1)}_2 &
=\tilde{B}^{(0)}_2-2(\tilde{w}^{(1)}_1)',\\
\tilde{B}^{(1)}_3 
&=\tilde{B}^{(0)}_3-3 (\tilde{w}^{(1)}_1)'\partial_x
-3(\tilde{w}^{(1)}_2)'+3\tilde{w}^{(1)}_1(\tilde{w}^{(1)}_1)'-3(\tilde{w}^{(1)}_1)''.
\end{aligned}
\end{equation}
{}From the compatibility of \eqref{SatoWilson-typeEq_W(0)}, 
\eqref{SatoWilson-typeEq_tildeB0n}, we obtain 
the Zhakharov-Shabat-type equations:
\begin{equation}
\frac{\partial\tilde{B}^{(l)}_i}{\partial t_j}-\frac{\partial\tilde{B}^{(l)}_j}{\partial t_i}
=\left[\tilde{B}^{(l)}_j,\,\tilde{B}^{(l)}_i\right] \quad (l=0,1,\;i,j=1,2,\ldots).
\label{ZakharovShabatEq_wEB}
\end{equation}
Now we set 
\begin{equation}
u_0=-\left(\tilde{w}^{(0)}_1\right)', \quad
u=-\left(\tilde{w}^{(1)}_1\right)'
\end{equation}
and consider \eqref{ZakharovShabatEq_wEB} with $l=0$, $i=2$, $j=3$, then it follows that
$u_0$ is a solution of the KP equation \eqref{KPeq} ($y=t_2$, $t=t_3$). 
The case $l=1$, $i=2$, $j=3$ of \eqref{ZakharovShabatEq_wEB} shows 
$u$ solves the KP equation with the background $u_0$:
\begin{equation}
\left\{4u_t-12uu_x-u_{xxx}-12\left(u_0u\right)_x\right\}_x -3u_{yy}=0.
\label{KP_with_u0}
\end{equation}
If $u_{yy}=0$ then \eqref{KP_with_u0} is reduced to 
the KdV equation with the background $u_0$:
\begin{equation}
4u_t-12uu_x-u_{xxx}-12\left(u_0u\right)_x=0, 
\label{KdV_with_u0}
\end{equation}
which has been discussed in \cite{Shabat}
by using integral operator formulation.
In this sense, the equations \eqref{SatoWilson-typeEq_W(0)}, 
\eqref{SatoWilson-typeEq_tildeB0n}, \eqref{ZakharovShabatEq_wEB} 
describe a hierarchy of soliton equations with the non-trivial background solution $u_0$.
We remark that $u_0=-\wp(x)$ when $\tilde{W}^{(0)}=V$. 
This case has been discussed in \cite{KuznetsovMikhailov} based on the method of \cite{Shabat}.

While we focus on the case where the $\wp$-function serves as the background solution 
in this paper, it may be possible that our approach could extend to 
deal with a broader range of background solutions, 
such as \cite{Eilbeck,Matsutani,NakayashikiIMRN2010}.

\section{A class of real-valued solutions}
{}From the viewpoint of applications in physics, real-valued and non-singular 
solutions are important. In this section, we will construct 
a class of real-valued solutions using the 
extended Li-Zhang solution \eqref{LiZhangSolution:KP:revisited}. 
We remark that real solutions are discussed also in \cite{Ichikawa} from 
algebro-geometric viewpoints. 

To construct real-valued solutions to the KP equation \eqref{KPeq}, 
we prepare a property of fundamental periods of the Weierstrass $\wp$-function. 
\begin{thm}[Pastras \cite{PastrasBook}, \S 3]
\label{thm:Pastras}
Let $g_2$, $g_3$ be the 
constants defined by \eqref{def:g2,g3}.
If the cubic polynomial $x^3-g_2x-g_3$ has three real roots, 
two fundamental half-periods $\omega_1$, $\omega_2$ can be 
selected so that one of them is real and the other purely imaginary.
\end{thm}
Hereafter in this section, we focus on the case of Theorem \ref{thm:Pastras}, 
and choose $\omega_1$ as real, $\omega_2$ as purely imaginary.
If $x$ is real, we have
\begin{equation}
\begin{aligned}
\overline{\wp\left(x+\omega_2\right)} &= \wp\left(x-\omega_2\right)=\wp\left(x+\omega_2\right),
\\
\overline{\zeta\left(x+\omega_2\right)} &= \zeta\left(x-\omega_2\right)
=\zeta\left(x+\omega_2\right)-2\zeta\left(\omega_2\right), 
\\
\overline{\sigma\left(x+\omega_2\right)} &= \sigma\left(x-\omega_2\right)
=-e^{-2\zeta\left(\omega_2\right)x}\sigma\left(x+\omega_2\right), 
\end{aligned}
\label{cc:wp,zeta,sigma}
\end{equation}
where we have used the (quasi-) periodic properties of the Weierstrass functions $\wp$, $\zeta$, $\sigma$:
\begin{equation}
\wp\left(z+2\omega_j\right) = \wp(z), \quad
\zeta\left(z+2\omega_j\right) = \zeta(z) + \zeta(\omega_j),\quad
\sigma\left(z+2\omega_j\right) = -e^{2\zeta\left(\omega_j\right)x}\sigma(z), 
\end{equation}
for $j=1,2$. 
Following \cite{KuznetsovMikhailov}, we introduce $\tilde{\Phi}$ as 
\begin{align}
\tilde{\Phi}(x;k)&=e^{\zeta(k)\omega_2-\zeta\left(\omega_2\right)k}\Phi\left(x+\omega_2;k\right)
=\frac{\sigma(x+\omega_2+k)}{\sigma(x+\omega_2)}e^{-\zeta(k)x-\zeta\left(\omega_2\right)k}. 
\end{align}
If $x$ and $k$ are real, one can show that 
\begin{equation}
\overline{\tilde{\Phi}(x;k)} = \tilde{\Phi}(x;k)
\end{equation}
by using \eqref{cc:wp,zeta,sigma}. 
Then the function
\begin{equation}
\begin{aligned}
\tilde{\Psi}(x,\bm{t};k) 
&= e^{\zeta(k)\omega_2-\zeta\left(\omega_2\right)k}
\Psi\left(x+\omega_2,\bm{t};k\right)\\
&= \tilde{\Phi}(x;k)\exp\!\left[\sum_{j=0}^\infty
\left\{\frac{\wp^{(j)}(k)}{(j+1)!}-\frac{a_j}{j+1}\right\}t_{j+2}\right]
\end{aligned}
\label{real-valued_PWF}
\end{equation}
is real-valued if all of $k$, $x$, $t_n$ ($n=1,2,\ldots$) are real.

Let $N,M$ be integers such that $N\leq M$. 
Following \cite{ChakravartyKodama,KodamaBook}, we introduce 
a set of real-valued eigenfunctions of $L_V$ of the form
\begin{equation}
\tilde{\varphi}_i =
\sum_{j=1}^M \tilde{a}_{i,j}\tilde{\Psi}\left(x,\bm{t};k_j\right) \quad (i=1,\ldots,N),
\label{real-valued-solution}
\end{equation}
where $\tilde{a}_{i,j}$ ($i=1,\ldots,N$, $j=1,\ldots,M$), 
$k_j$ ($j=1,\ldots,M$) are real constants.
Furthermore we assume $N\times M$-matrix $\tilde{A}=
\left[\tilde{a}_{i,j}\right]^{i=1,\ldots,N}_{j=1,\ldots,M}$ is 
totally nonnegative 
\cite{ChakravartyKodama,KodamaBook}. 

If we set $M=2N$ and assume the following conditions, 
\begin{equation}
\begin{aligned}
& k_{2j-1}+k_{2j}=0 \quad (j=1,\ldots,N), \\
& \tilde{a}_{i,j}=0 \text{ \ if \ } j\not\in\left\{2i-1,2i\right\},
\end{aligned}
\end{equation}
then $\tilde{\varphi}_i$ of \eqref{real-valued-solution} 
is reduced to the following form:
\begin{equation}
\tilde{\varphi}_i 
= \tilde{a}_{i,2i-1}\tilde{\Psi}\left(x,\bm{t};k_{2i-1}\right)
+\tilde{a}_{i,2i}\tilde{\Psi}\left(x,\bm{t};-k_{2i-1}\right)
\quad (i=1,\ldots,N).
\end{equation}
Since $\wp^{(2n)}(x)$ is an even function, it follows that
\begin{equation}
\partial_{2n}\tilde{\varphi}_i 
= B_{2n}^{(V)}\tilde{\varphi}_i
= \left(\tilde{P}_{2n}-\frac{a_{n-2}}{n-1}\right)\tilde{\varphi}_i
= \left(\frac{\wp^{(2n-2)}(k_{2i-1})}{(2n-1)!}-\frac{a_{n-2}}{n-1}\right)\tilde{\varphi}_i
\quad (i=1,\ldots,N).
\end{equation}
This means $\tilde{\varphi}_i$ satisfies the condition \eqref{reducedEigenFcn}. 
Thus we can apply Lemma \ref{lemma:EDT_reducedCase} to the stationary solution $L_V$.
In this way we can reproduce the Li-Zhang solution \eqref{LiZhangWronskian}. 

To construct concrete examples, we consider the lemniscate case 
(\cite{AbramowitzStegun}, \S 18.14): 
\begin{equation}
g_2=1, \quad g_3=0, \quad 
\omega_1=\frac{1}{4\sqrt{\pi}}\,\Gamma\!\left(\frac{1}{4}\right)^2
=1.854074677\ldots, \quad
\omega_2=\sqrt{-1}\,\omega_1,
\end{equation}
where $\Gamma(\,\cdot\,)$ is the Gamma function.
In the case of $N=1$, $M=2$ in \eqref{real-valued-solution}, 
the generalized Li-Zhang solution \eqref{LiZhangSolution:KP:revisited} now takes the form
\begin{equation}
u = -\wp(x) +\partial_x^2\log\left[
a_{1,1}\tilde{\Psi}\left(x,\bm{t};k_1\right)
+a_{1,2}\tilde{\Psi}\left(x,\bm{t};k_2\right)\right], 
\end{equation}
which corresponds to a line-soliton solution that interacts the background elliptic wave
(Figure \ref{fig:1-line-soliton}). 
In this paper, figures are drawn with Wolfram Mathematica 13.3.
In Figure \ref{fig:1-line-soliton}, 
one can observe the phase-shift of the background wave, which is due to the 
interaction between the background wave and the line-soliton.
In the case of $N=1$, $M=3$, the corresponding solution exhibits resonant interaction
(Figure \ref{fig:line-soliton-resonance}).
The case $N=2$, $M=4$ with
\begin{equation}
\tilde{A}_0 = \begin{bmatrix}
1 & 1 & 0 & 0\\
0 & 0 & 1 & 1
\end{bmatrix}
\label{Amatrix:O-type}
\end{equation}
is shown in Figure \ref{fig:2-line-solitons}, 
which describes interaction between two line-solitons and background wave. 
The choice \eqref{Amatrix:O-type} is an example of ``O-type'' in the sense of 
\cite{ChakravartyKodama,KodamaBook}. 
It should be remarked that similar figures have been presented 
in \cite{Nakayashiki_arXiv2023}, 
which are associated with 
singular algebraic curves (cf. \cite{Ichikawa,NakayashikiGMP2020}). 

If we choose $\tilde{A}_1$, $\tilde{A}_2$ as 
\begin{equation}
\tilde{A}_1 = \begin{bmatrix}
 1 & 1 & 0 & -0.6 \\
 0 & 0 & 1 & 1 
\end{bmatrix}, \quad
\tilde{A}_2 = \begin{bmatrix}
 1 & 1 & -0.4 & -0.6 \\
 0 & 0 & 1 & 1 
\end{bmatrix}, 
\label{tildeA1A2}
\end{equation}
more complicated patterns like Figure \ref{fig:24soliton_1},
Figure \ref{fig:24soliton_2} can be observed. 
It is expected that more complex web-like patters, 
like those described in \cite{ChakravartyKodama,IsojimaWilloxSatsuma,KodamaBook,MarunoBiondini}, 
can be observed with suitably chosen parameters. 

\begin{figure}[p]
\begin{center}
\begin{tabular}{ccc}
\raisebox{3cm}{$t=-2$:} & \includegraphics[scale=0.35]{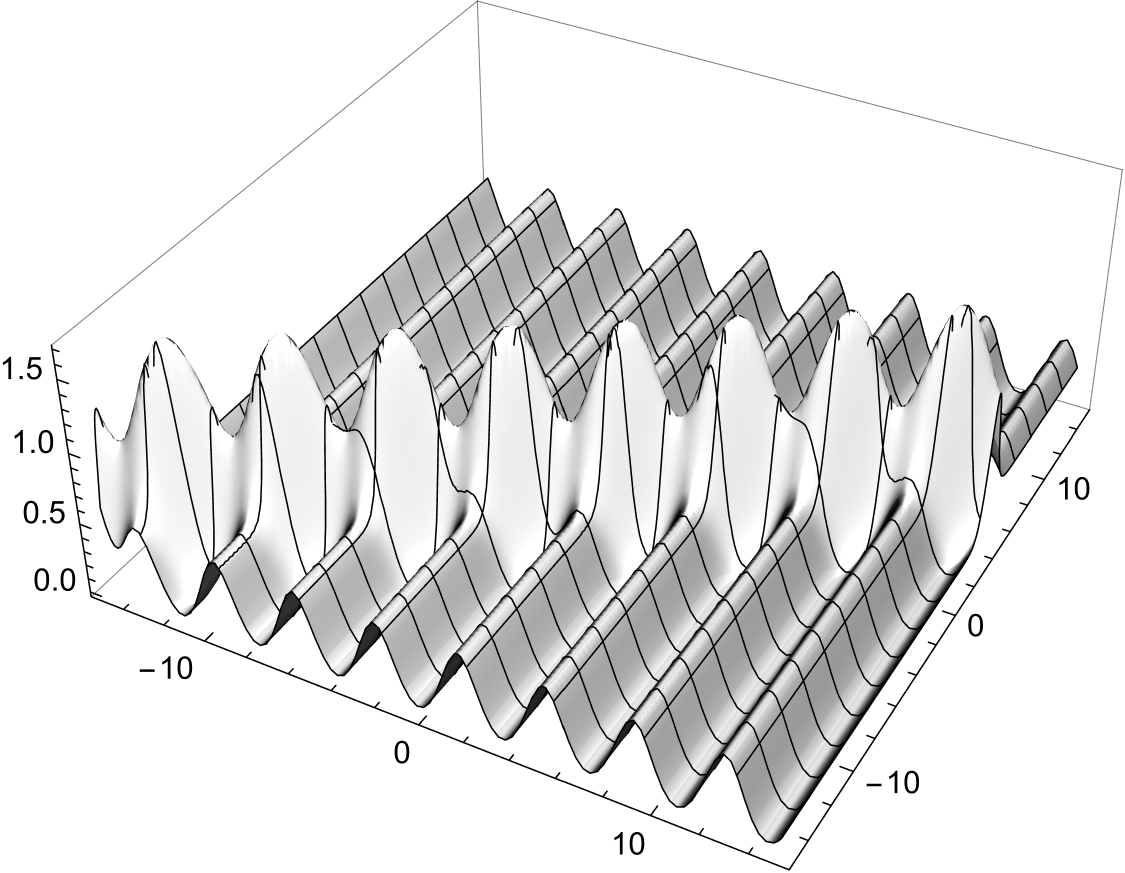}
 & \includegraphics[scale=0.26]{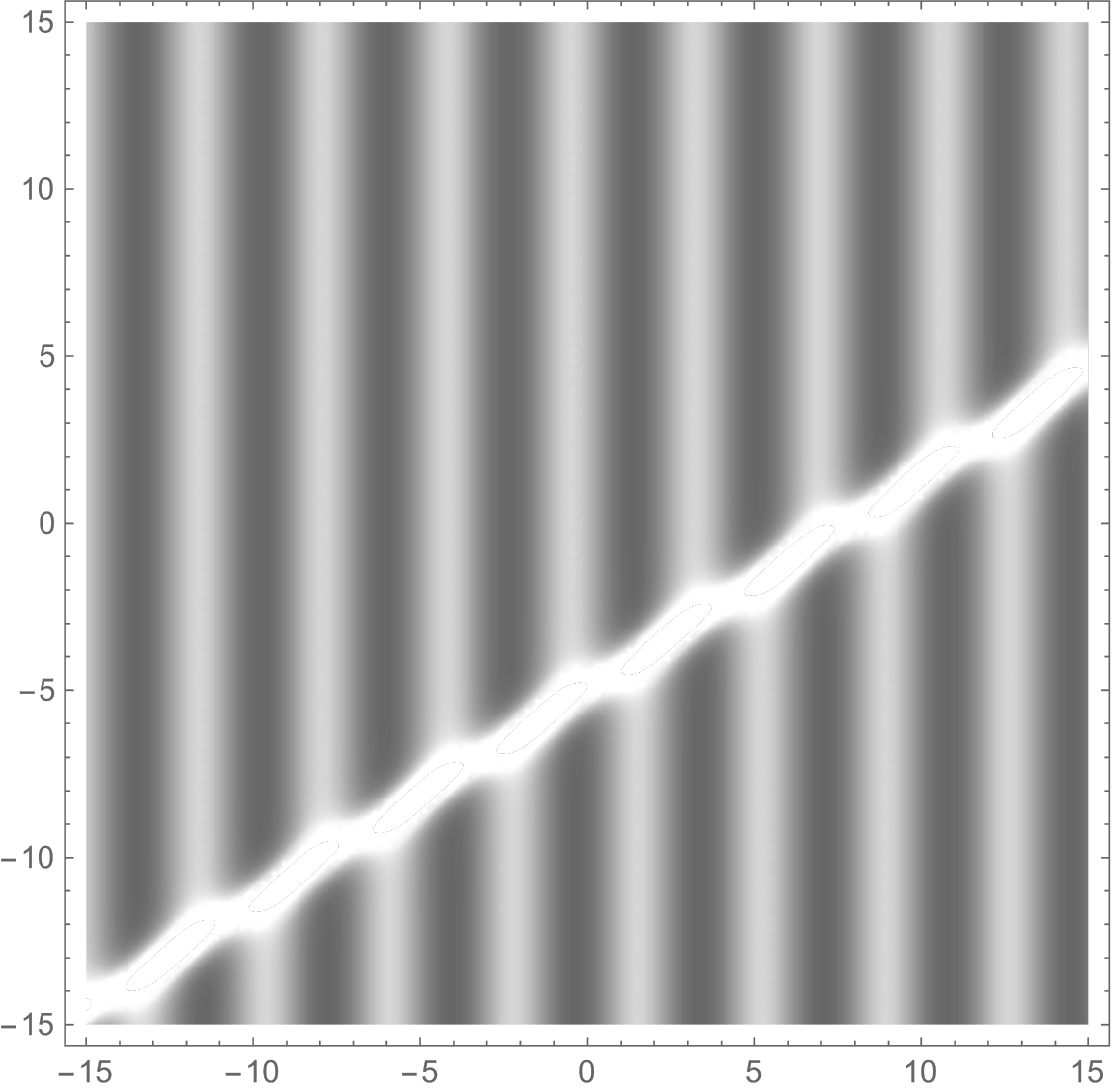}\\
\raisebox{3cm}{$t=2$:} & \includegraphics[scale=0.35]{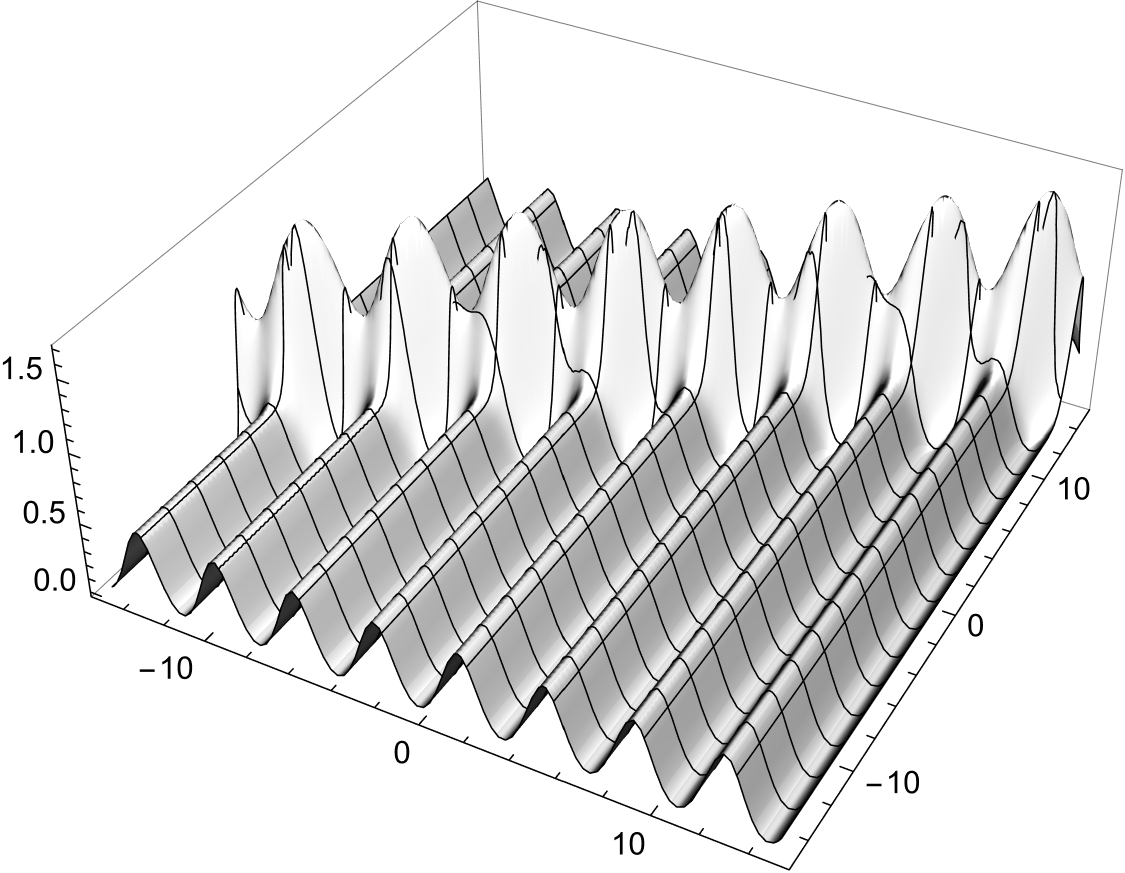}
 & \includegraphics[scale=0.26]{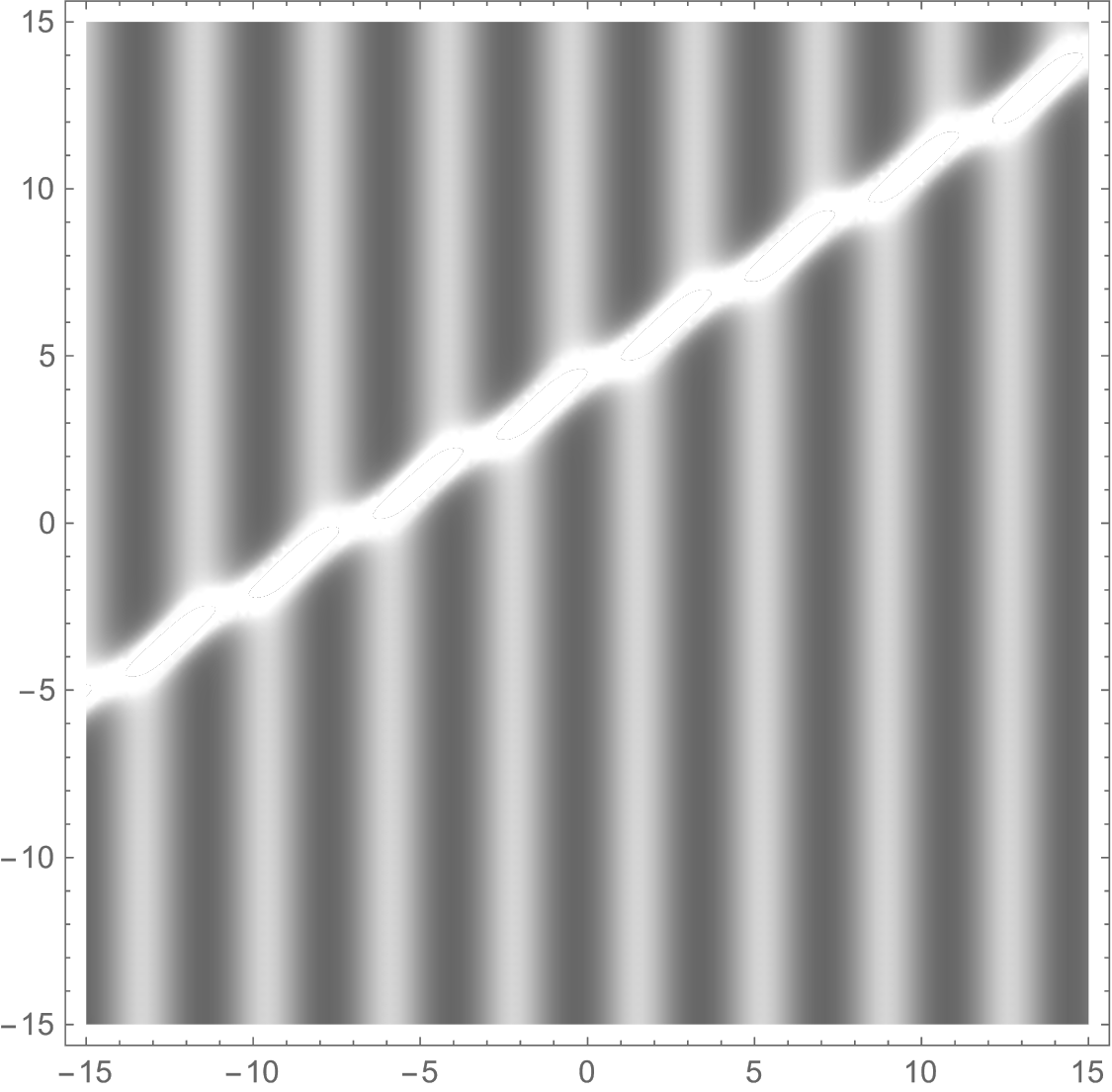}
\end{tabular}
\end{center}
\caption{Line-soliton propagation with the elliptic background $u_0=-\wp(x)$
(the lemniscate case, $N=1$, $M=2$, $\tilde{a}_{1,1}=\tilde{a}_{1,2}=1$, $k_1=-3.2$, $k_2=-1.5$). 
In the graphs in the right column, the brighter the area, the larger the value of $u$.}
\label{fig:1-line-soliton}
\end{figure}

\begin{figure}[p]
\begin{center}
\begin{tabular}{ccc}
\raisebox{3cm}{$t=-2$:} & \includegraphics[scale=0.35]{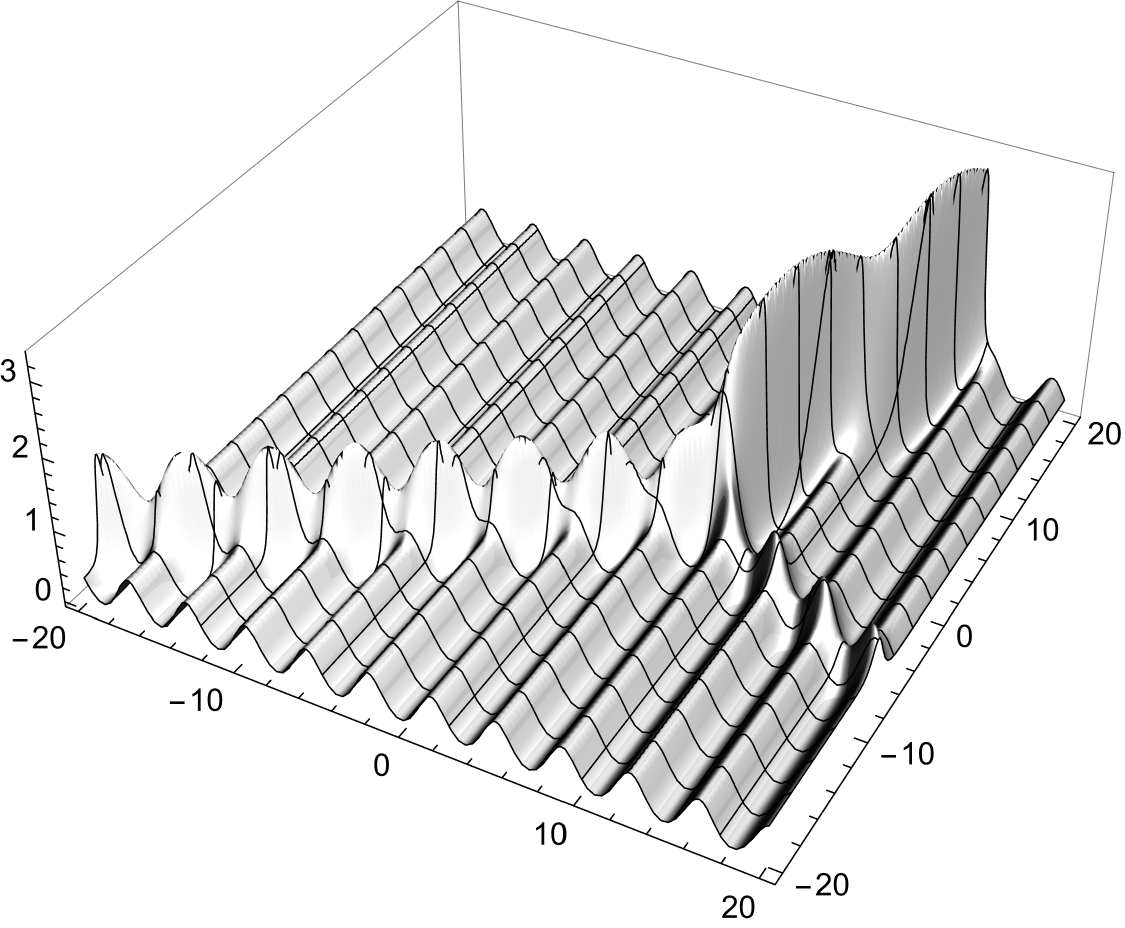}
 & \includegraphics[scale=0.26]{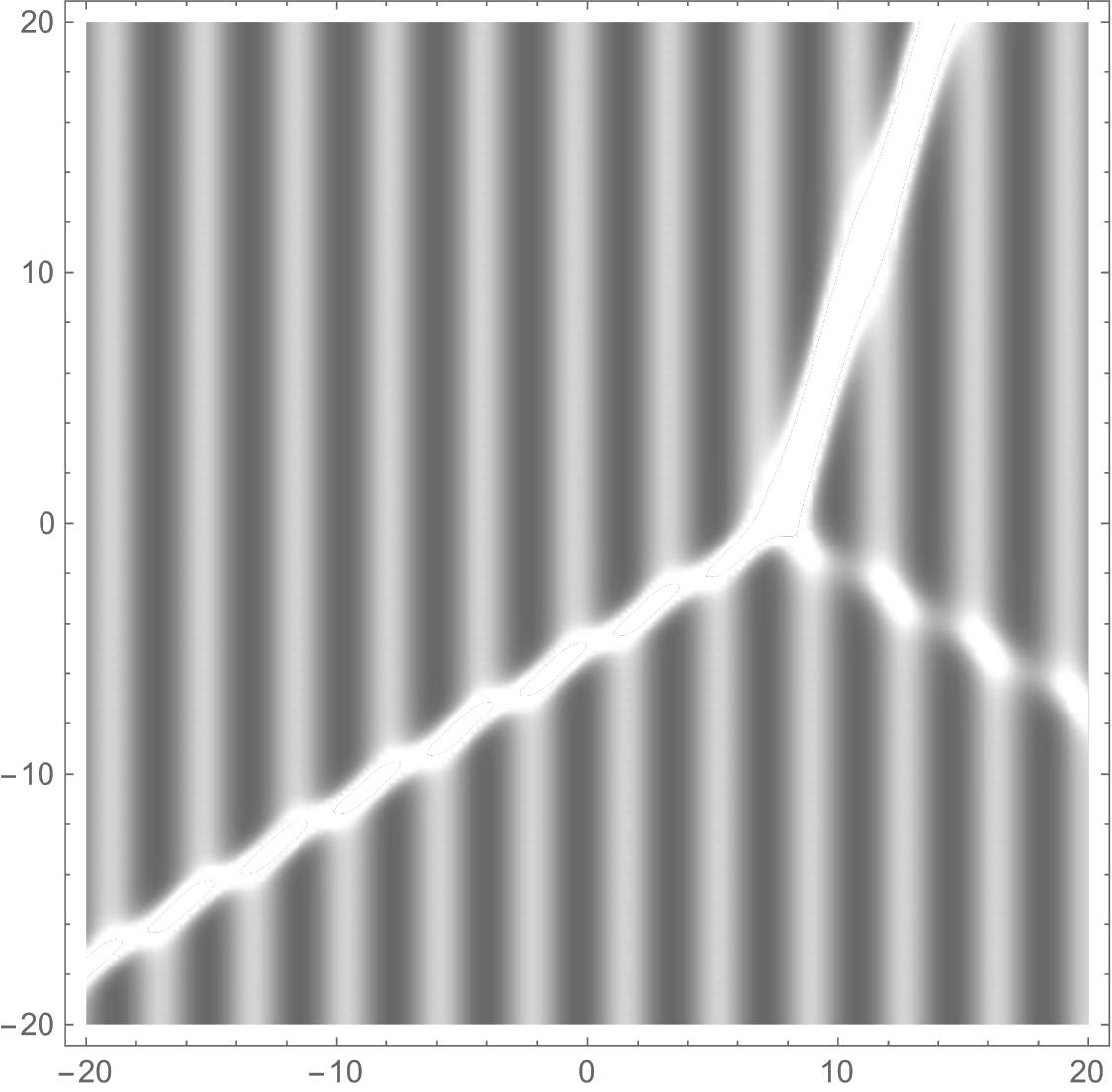}\\
\raisebox{3cm}{$t=2$:} & \includegraphics[scale=0.35]{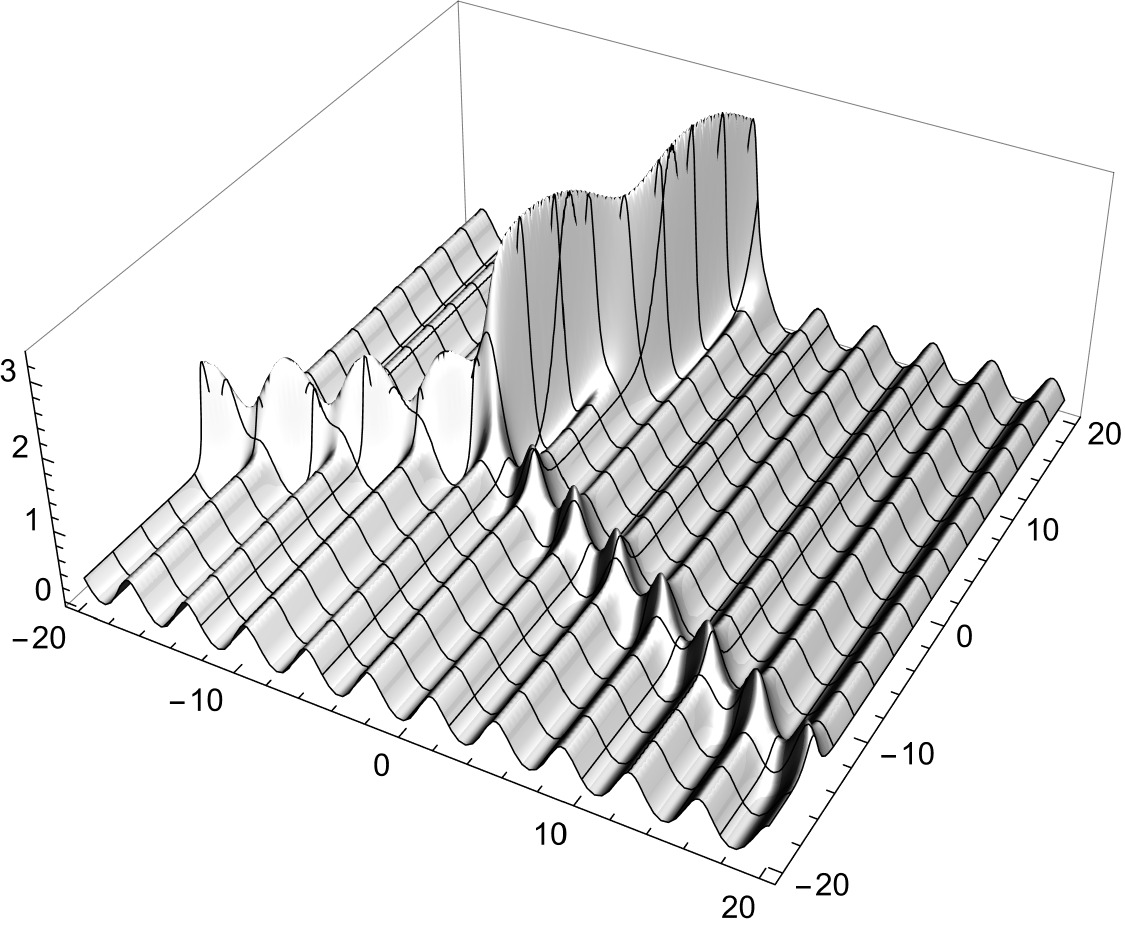}
 & \includegraphics[scale=0.26]{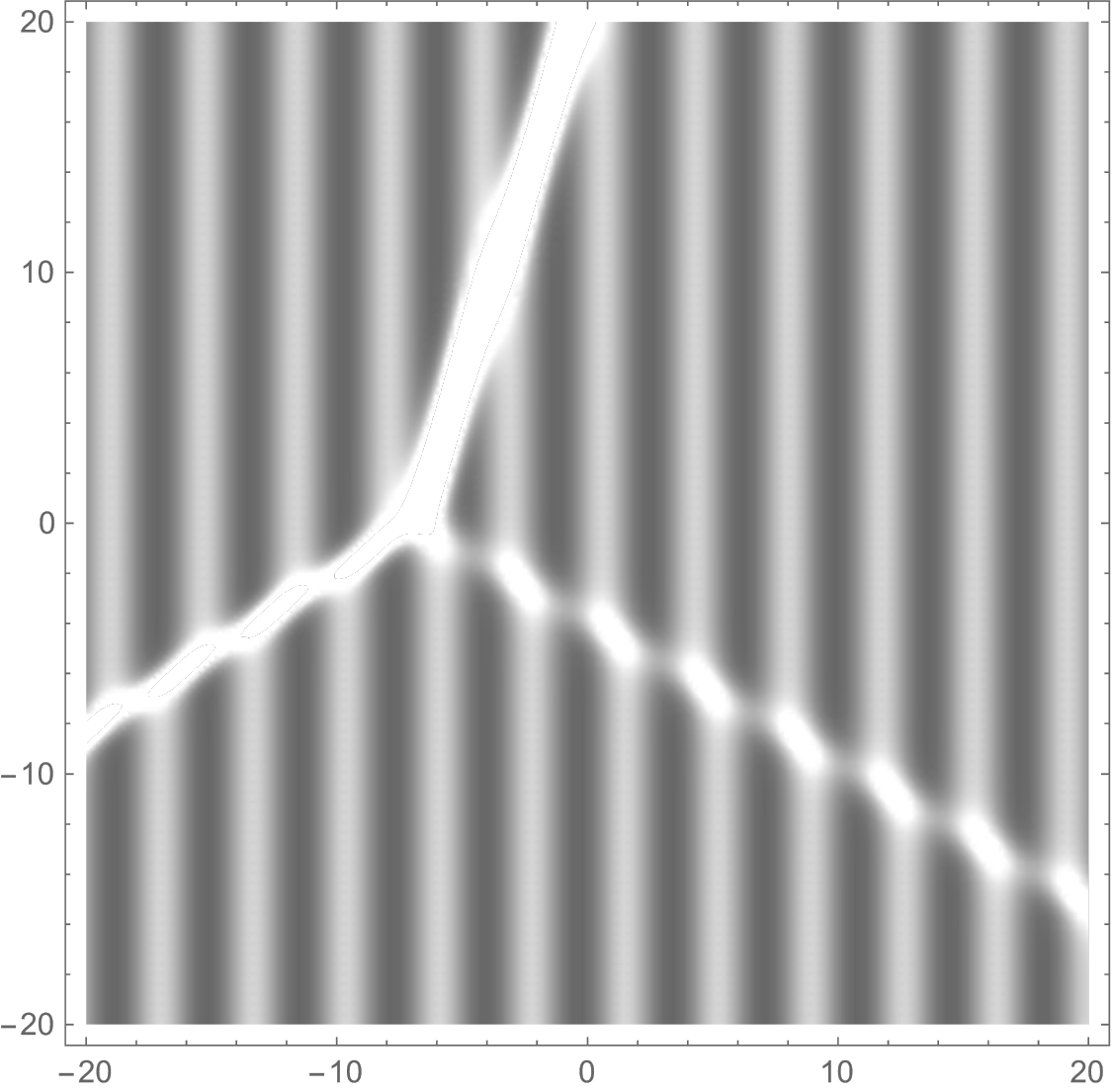}
\end{tabular}
\end{center}
\caption{Resonant interaction of line-solitons with the elliptic background $u_0=-\wp(x)$
(the lemniscate case, $N=1$, $M=3$, $\tilde{a}_{1,1}=\tilde{a}_{1,2}=\tilde{a}_{1,3}=1$, 
$k_1=-3.2$, $k_2=-1.5$, $k_3=-0.6$).
In the graphs in the right column, the brighter the area, the larger the value of $u$.}
\label{fig:line-soliton-resonance}
\end{figure}

\begin{figure}[p]
\begin{center}
\begin{tabular}{ccc}
\raisebox{3cm}{$t=-1$:} & \includegraphics[scale=0.35]{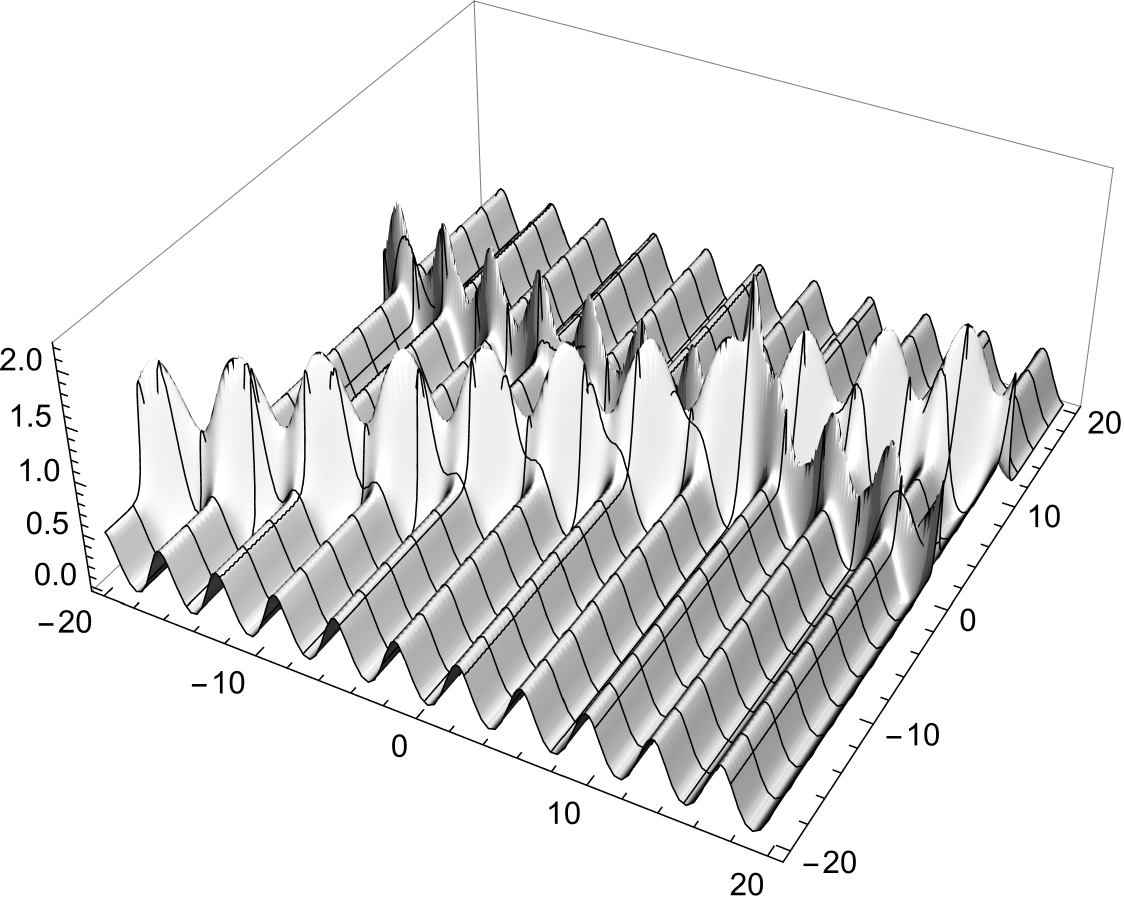}
 & \includegraphics[scale=0.26]{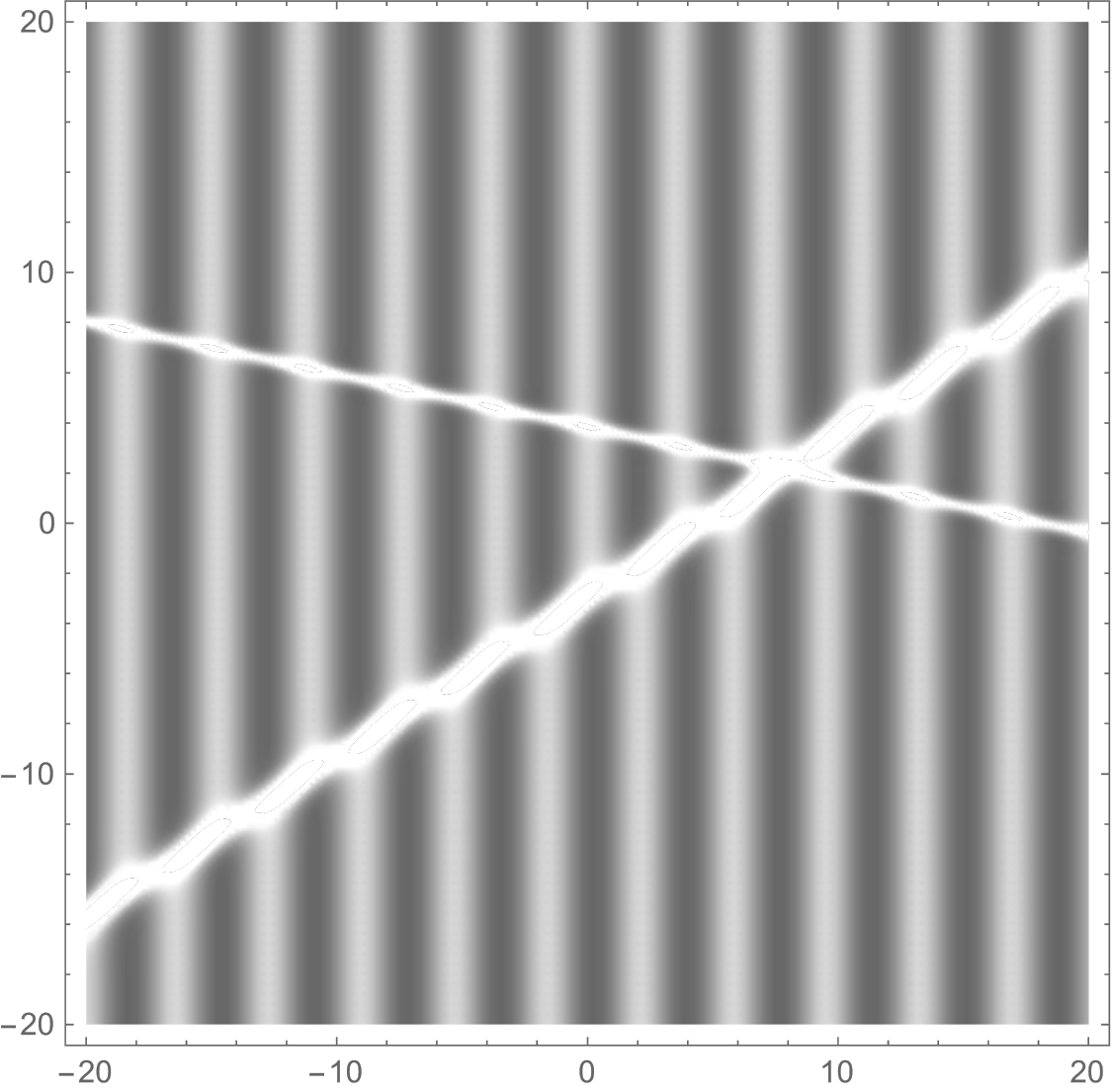}\\
\raisebox{3cm}{$t=1$:} & \includegraphics[scale=0.35]{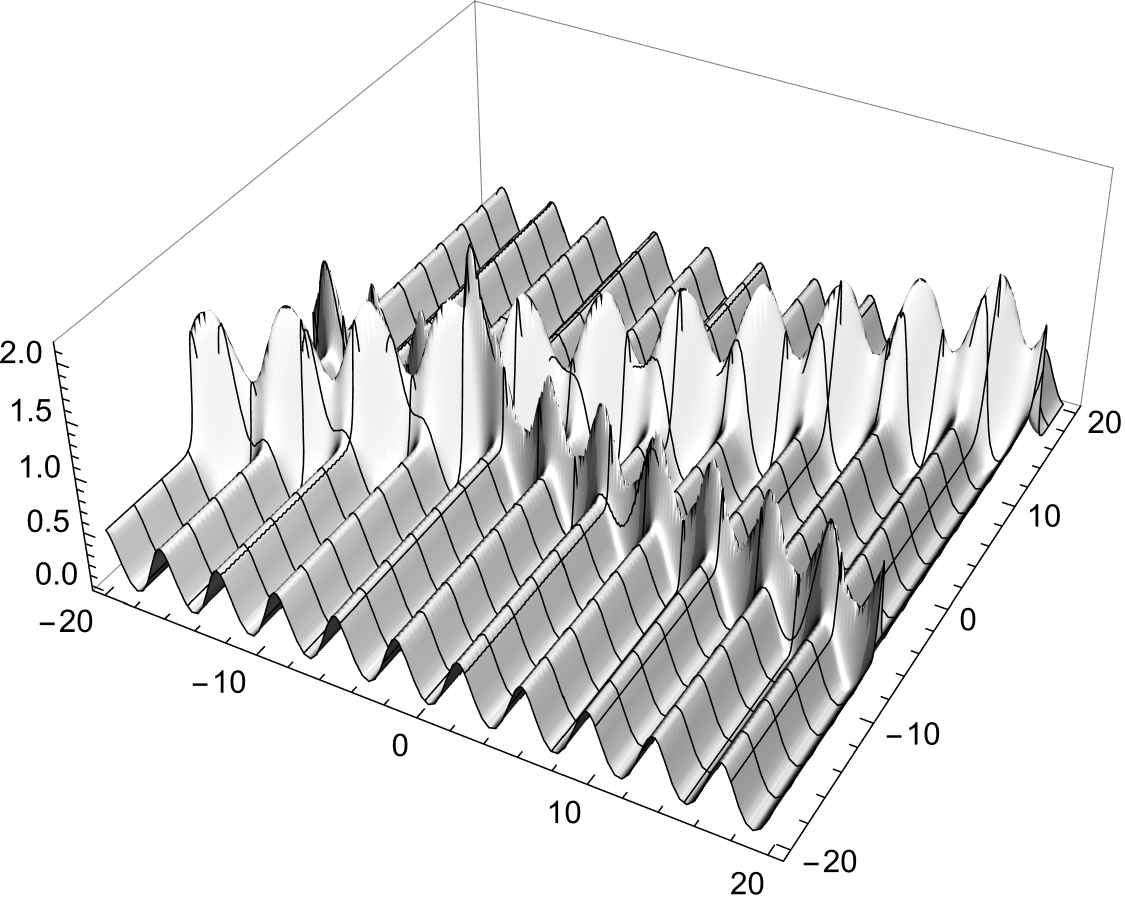}
 & \includegraphics[scale=0.26]{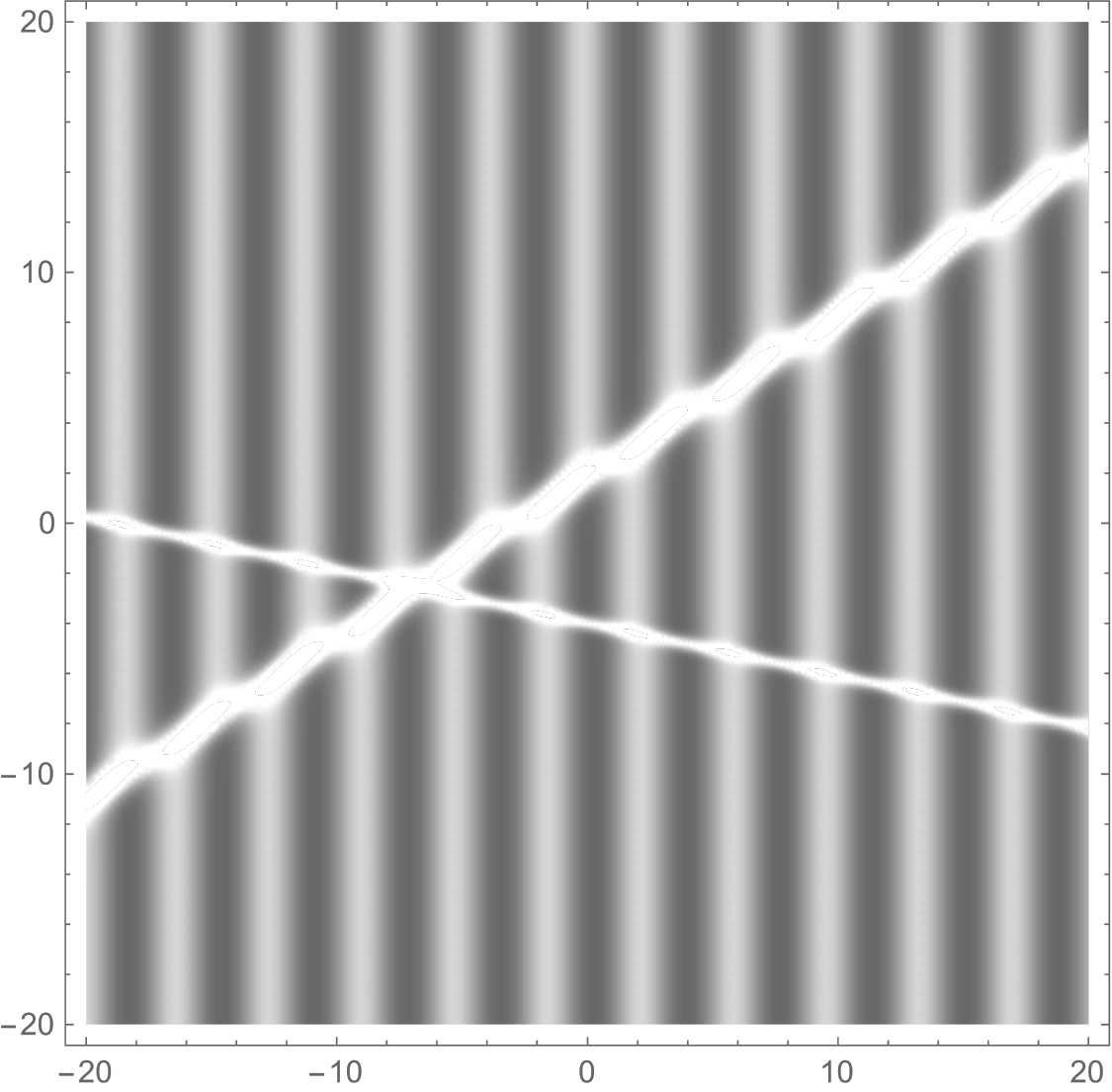}
\end{tabular}
\end{center}
\caption{Interaction between two line-solitons with the elliptic background $u_0=-\wp(x)$
(the lemniscate case, $N=2$, $M=4$, 
$\tilde{A}$ is chosen as \eqref{Amatrix:O-type}, 
$k_1=-3.2$, $k_2=-1.5$, $k_3=-0.6$, $k_4 = -0.3$). 
In the graphs in the right column, the brighter the area, the larger the value of $u$.}
\label{fig:2-line-solitons}
\end{figure}

\begin{figure}[p]
\begin{center}
\begin{tabular}{ccc}
\raisebox{3cm}{$t=-2$:} & \includegraphics[scale=0.35]{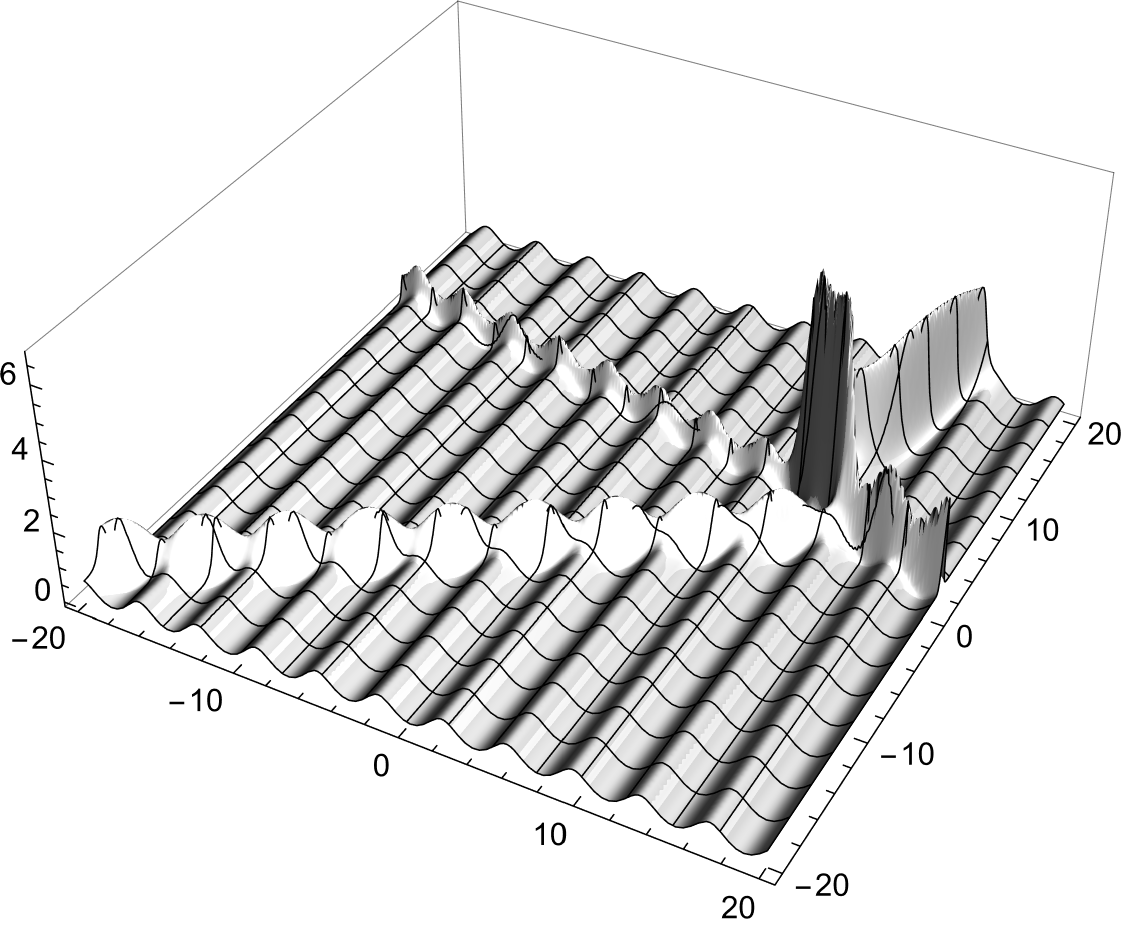}
 & \includegraphics[scale=0.26]{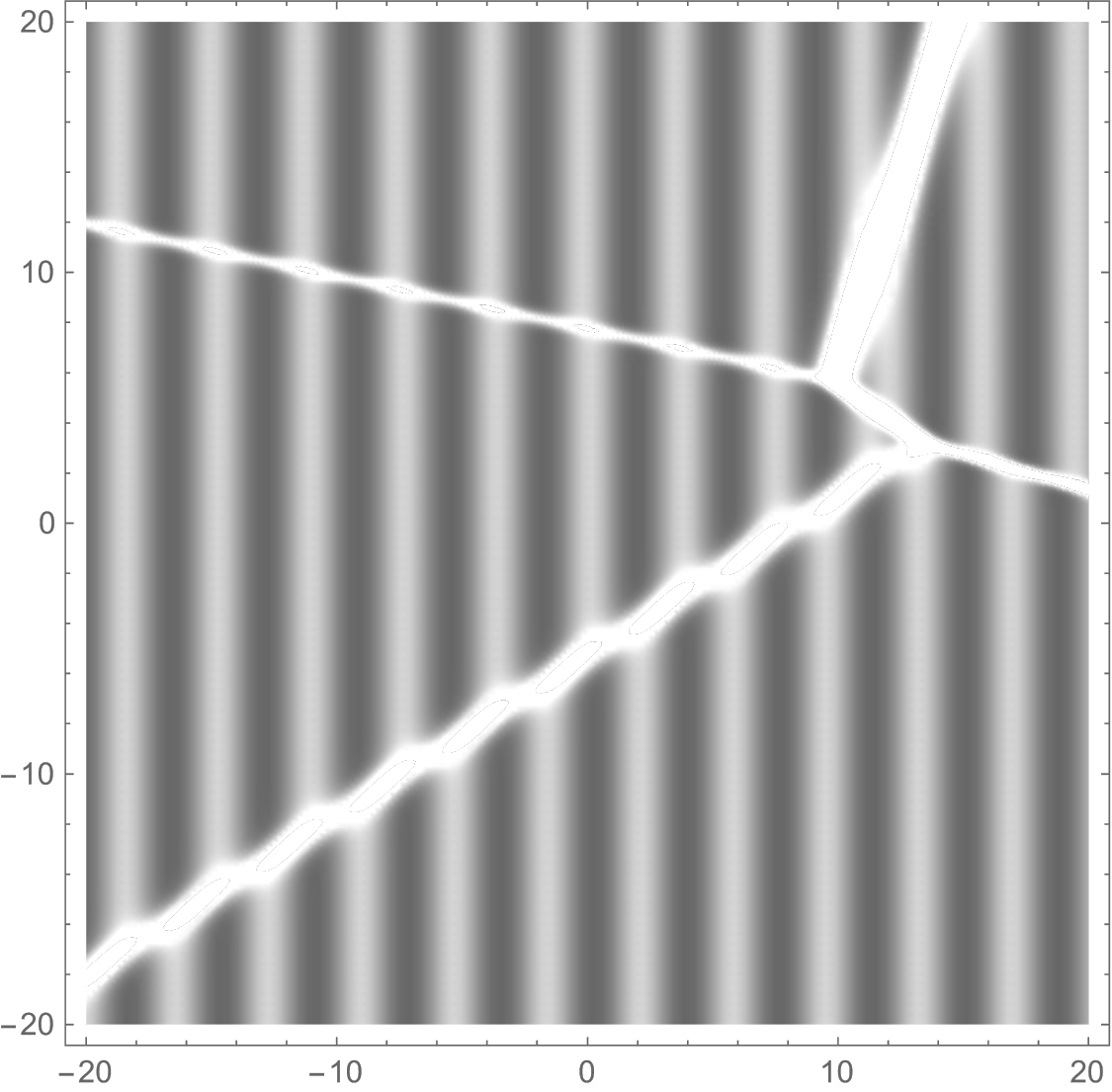}\\
\raisebox{3cm}{$t=0$:} & \includegraphics[scale=0.35]{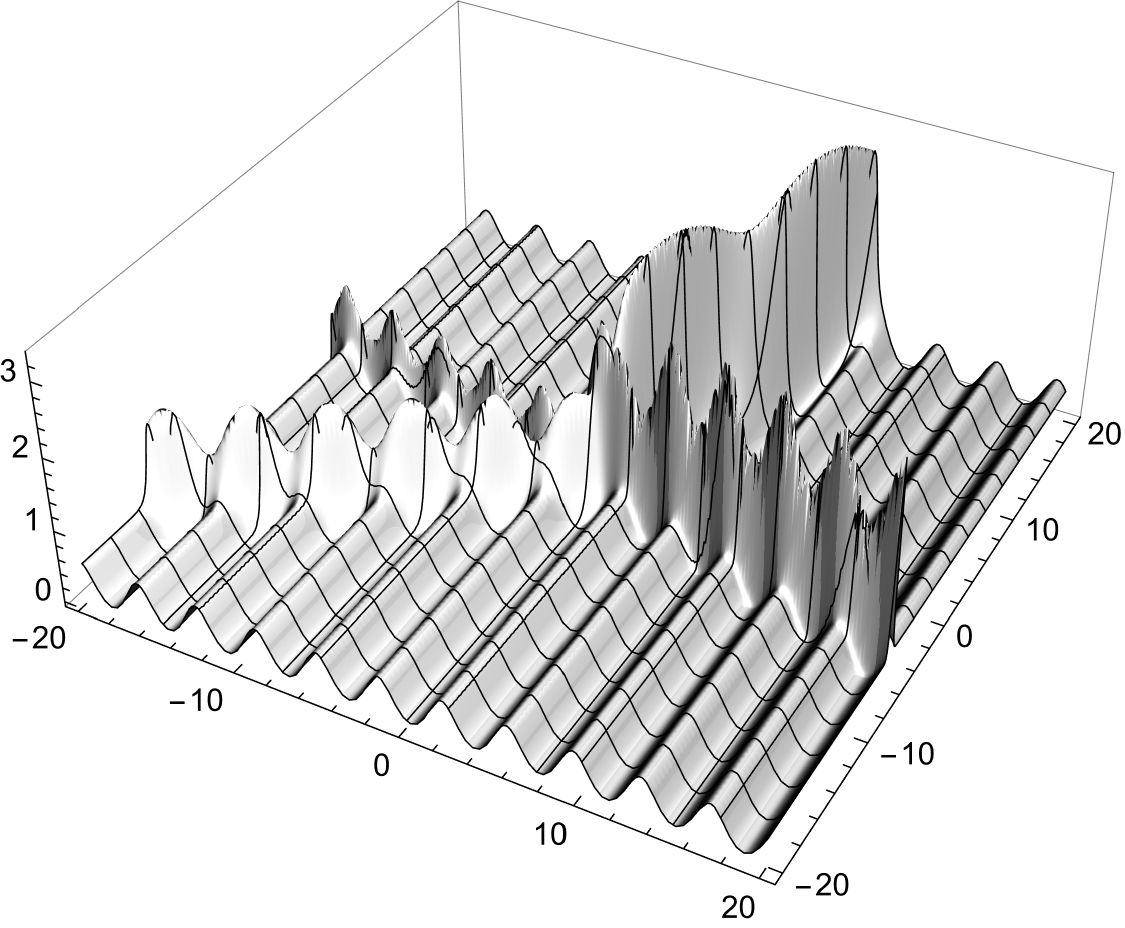}
 & \includegraphics[scale=0.26]{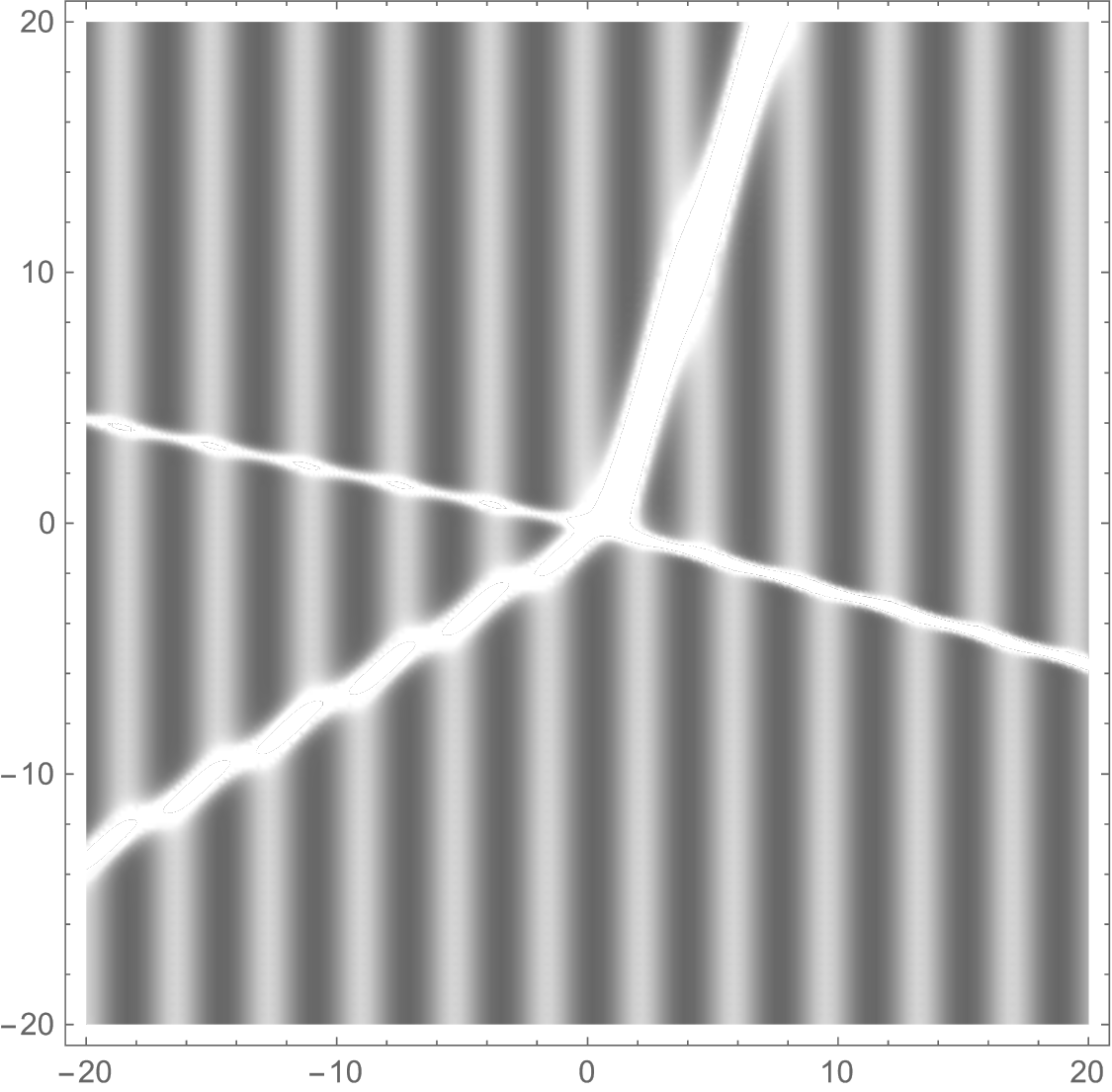}\\
\raisebox{3cm}{$t=2$:} & \includegraphics[scale=0.35]{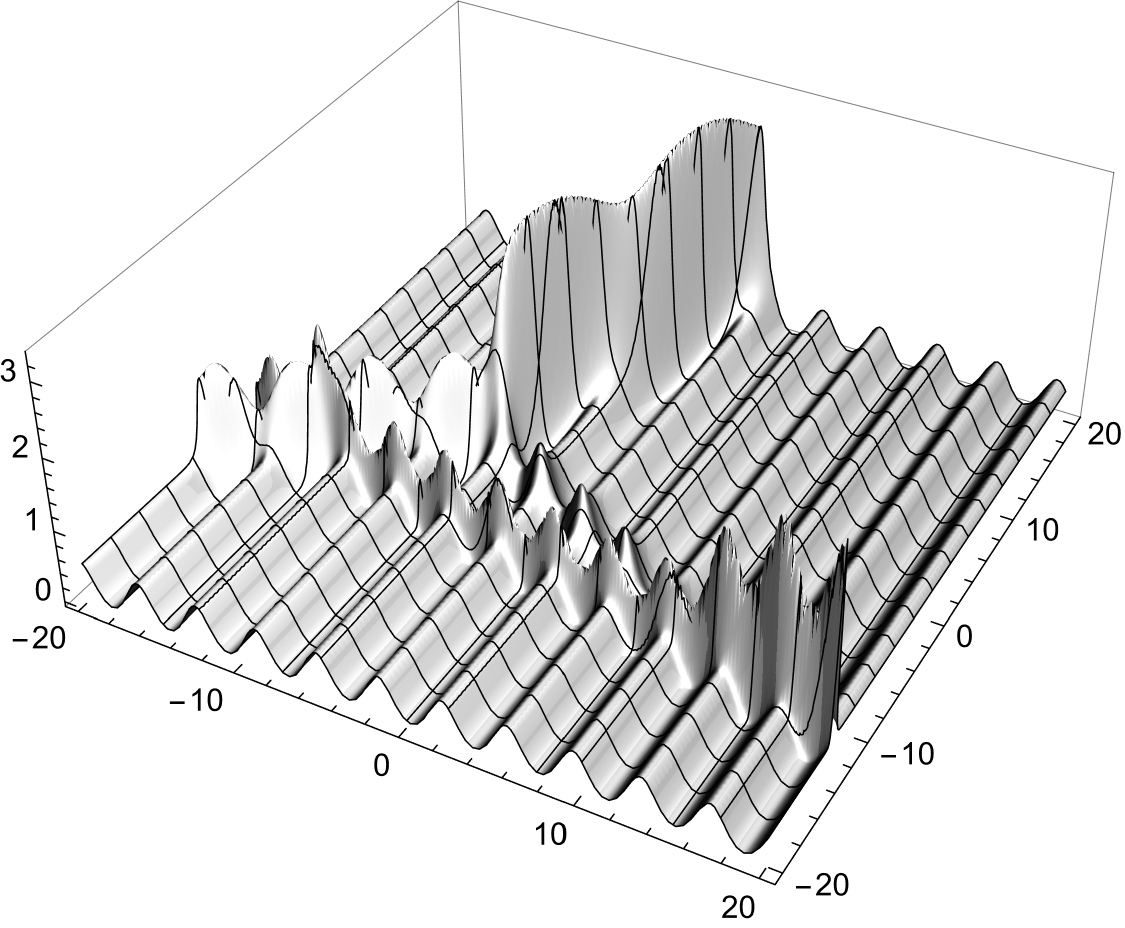}
 & \includegraphics[scale=0.26]{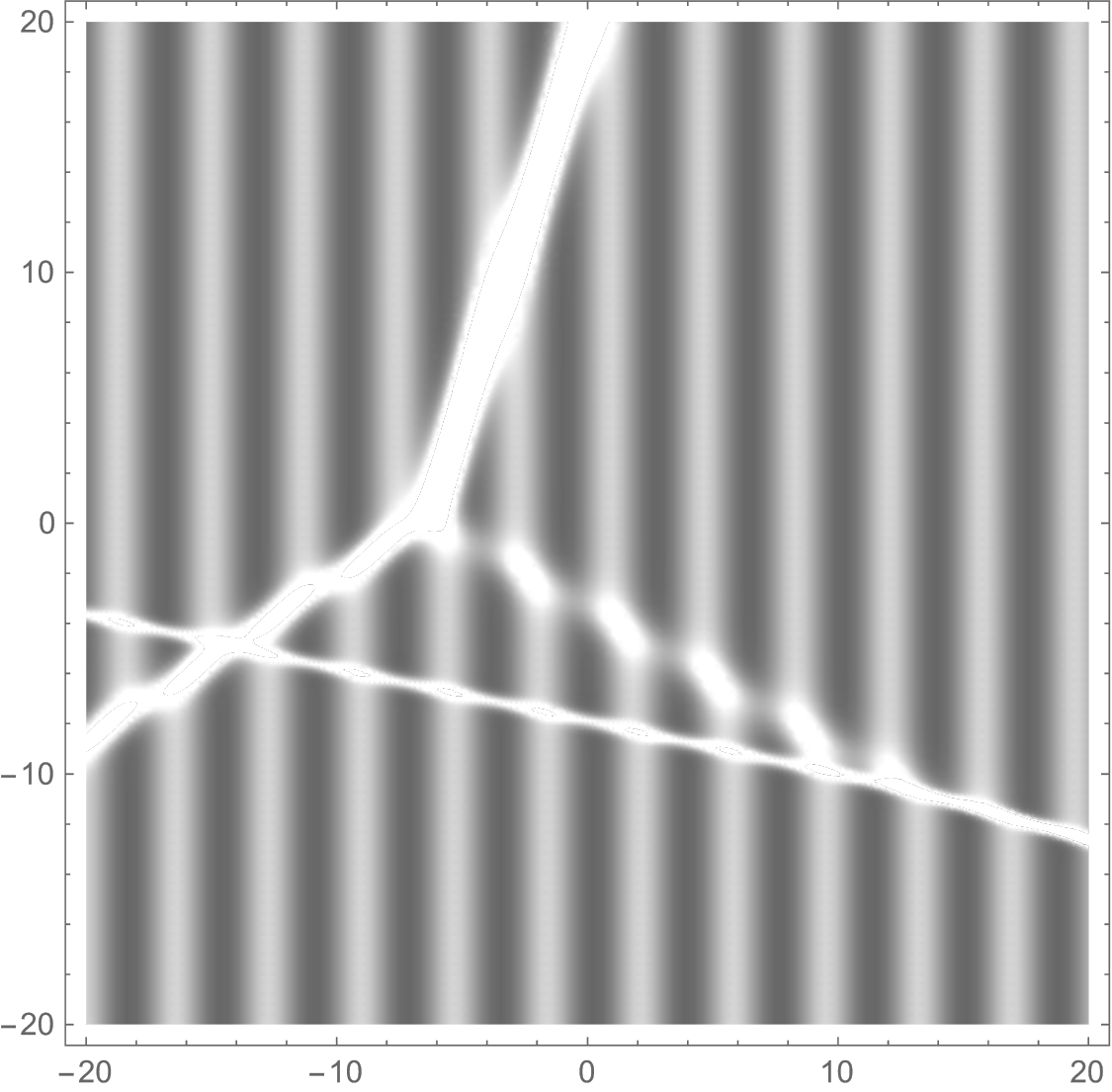}
\end{tabular}
\end{center}
\caption{Example of complex interaction 1 (the lemniscate case, 
$N=2$, $M=4$, $\tilde{A}$ is chosen as $\tilde{A}_1$ of \eqref{tildeA1A2}, 
$k_1=-3.2$, $k_2=-1.5$, $k_3=-0.6$, $k_4 = -0.3$).}
\label{fig:24soliton_1}
\end{figure}

\begin{figure}[p]
\begin{center}
\begin{tabular}{ccc}
\raisebox{3cm}{$t=-2$:} & \includegraphics[scale=0.35]{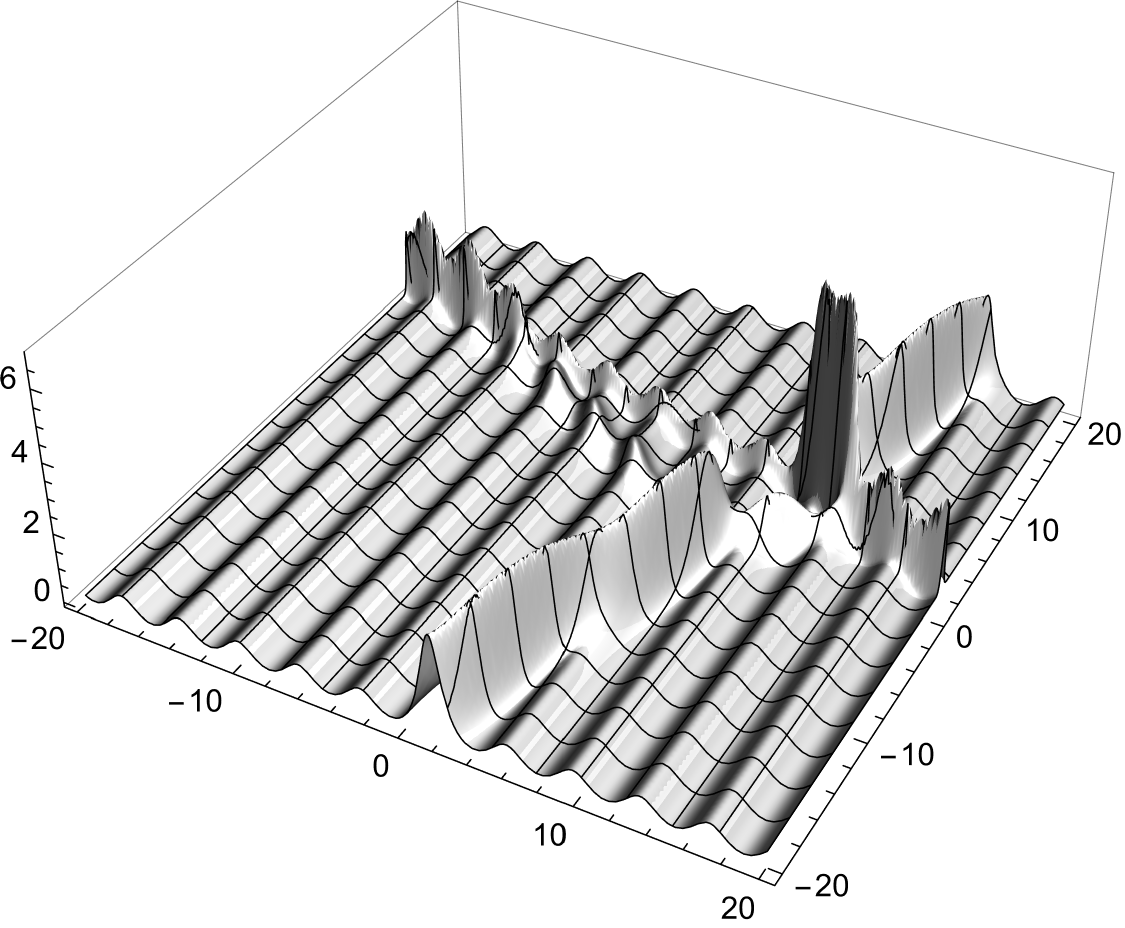}
 & \includegraphics[scale=0.26]{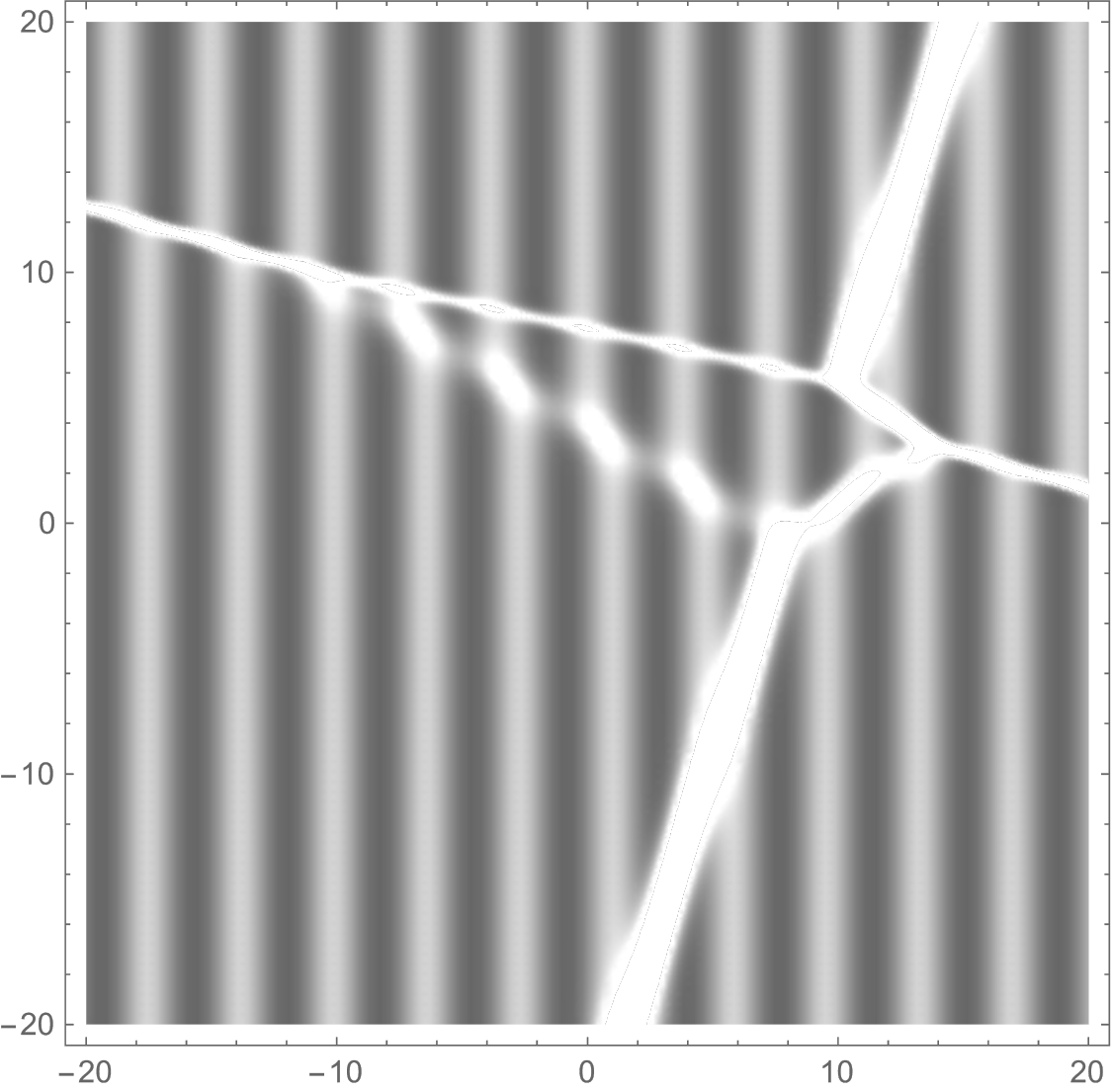}\\
\raisebox{3cm}{$t=0$:} & \includegraphics[scale=0.35]{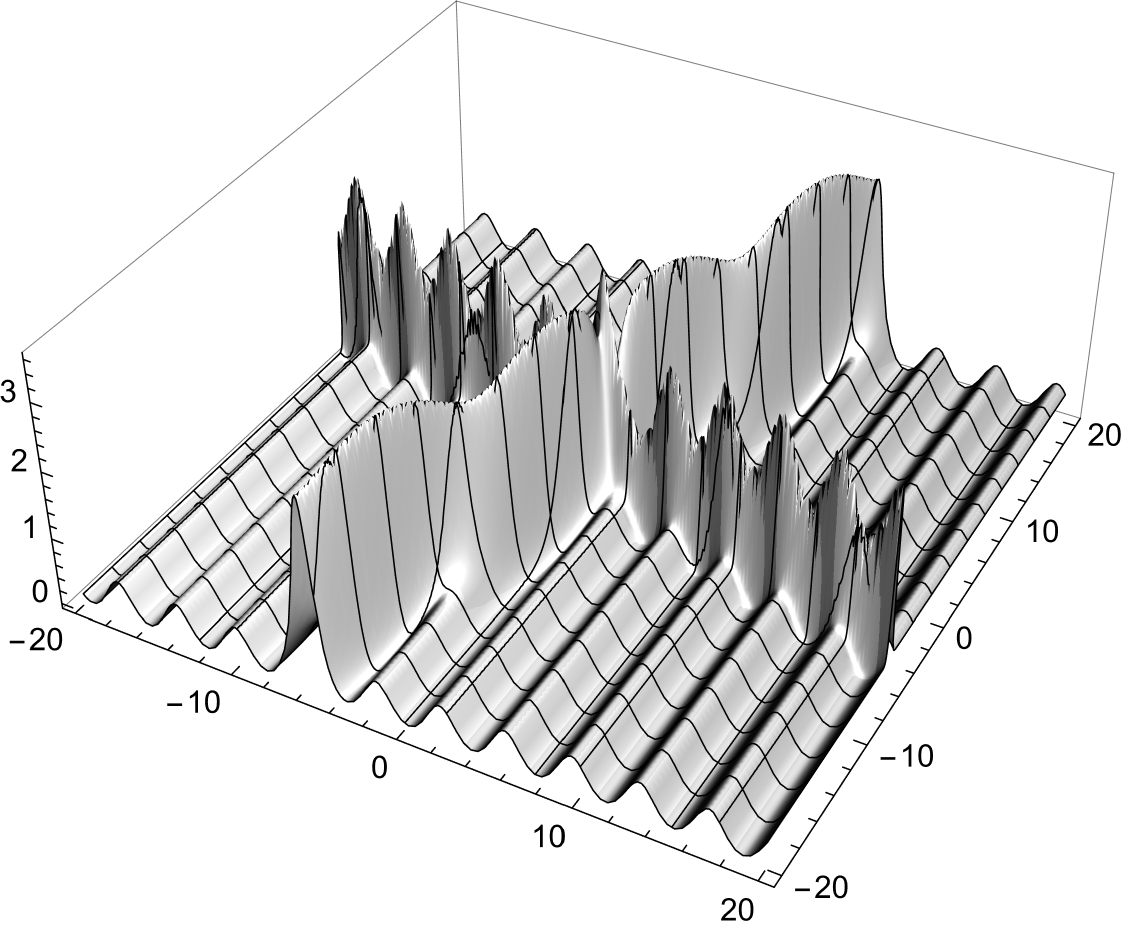}
 & \includegraphics[scale=0.26]{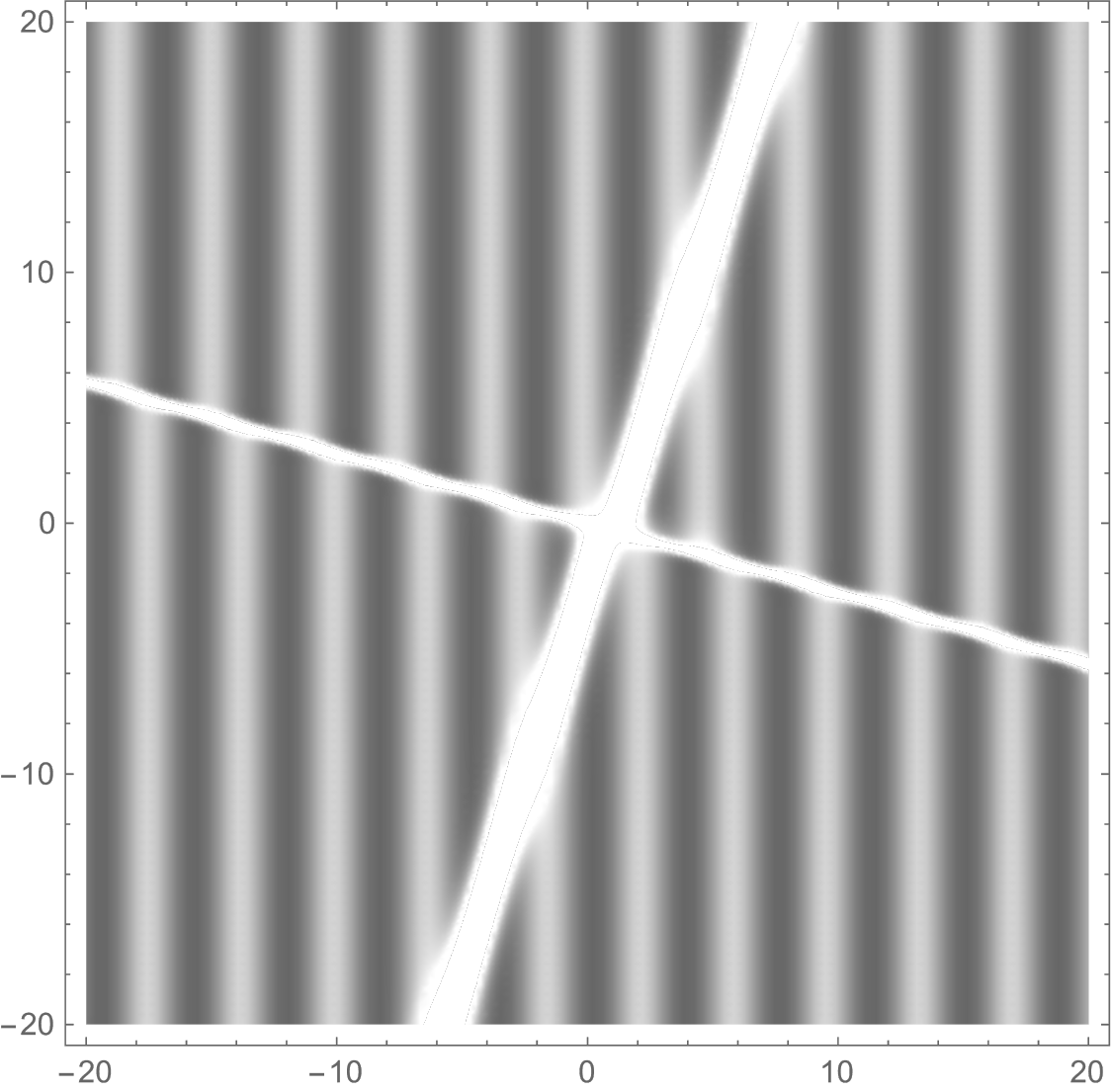}\\
\raisebox{3cm}{$t=2$:} & \includegraphics[scale=0.35]{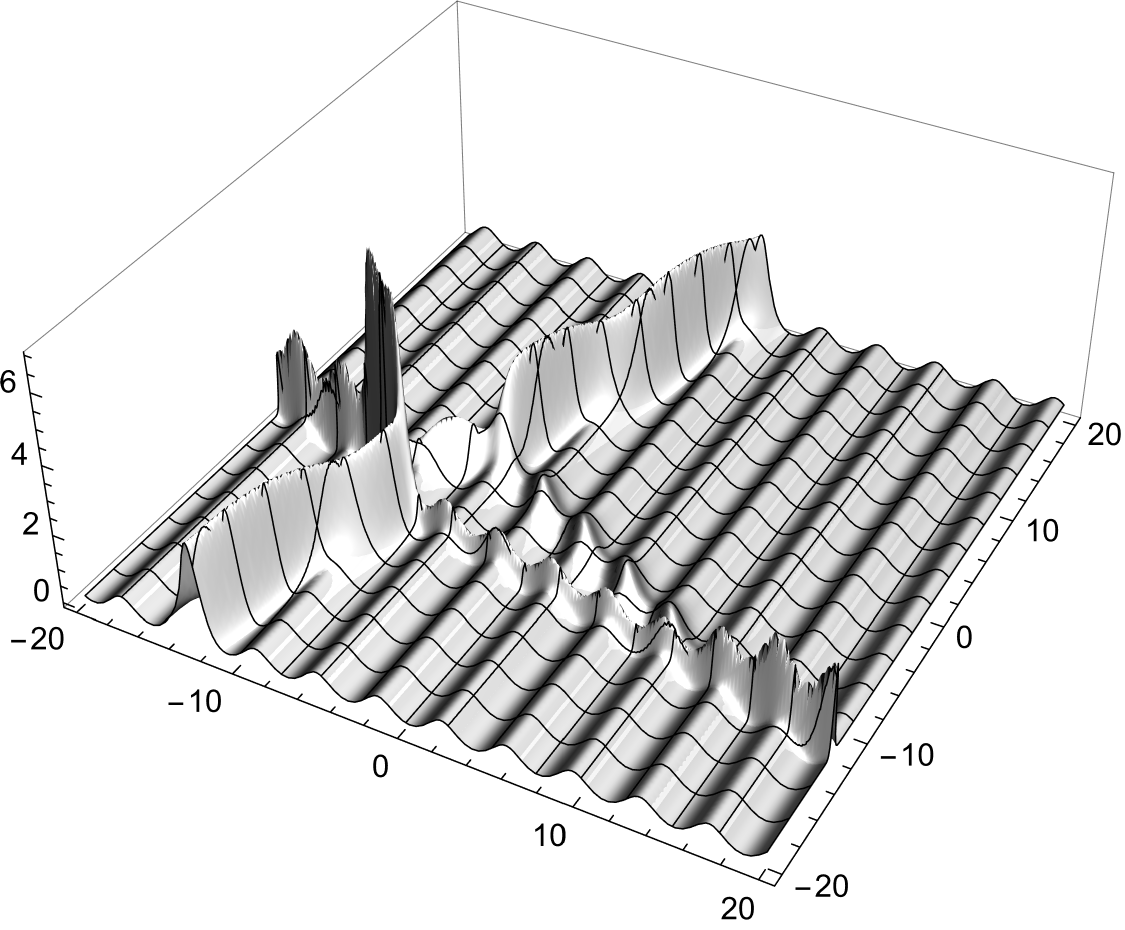}
 & \includegraphics[scale=0.26]{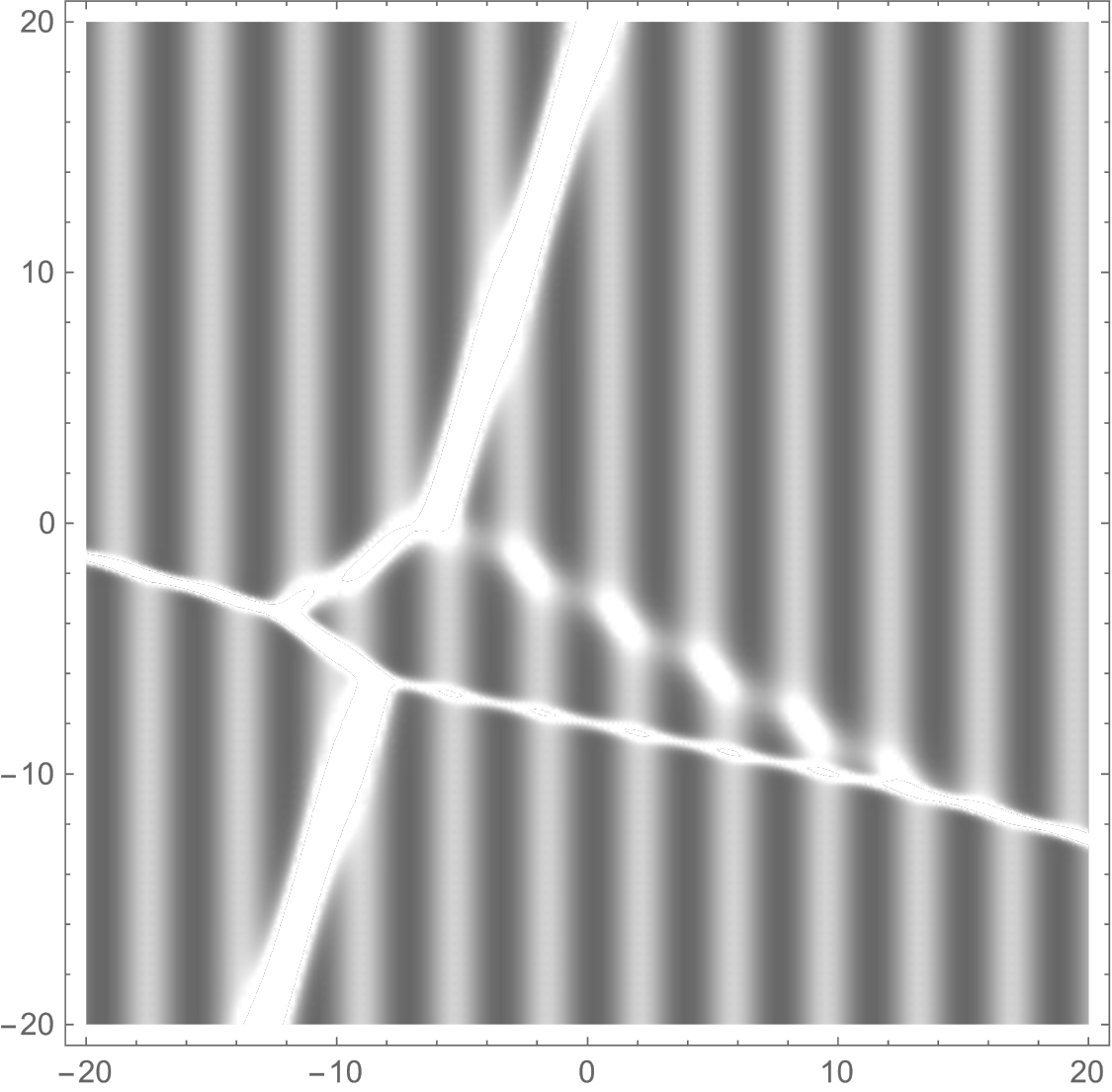}
\end{tabular}
\end{center}
\caption{Example of complex interaction 2 (the lemniscate case, 
$N=2$, $M=4$, $\tilde{A}$ is chosen as $\tilde{A}_2$ of \eqref{tildeA1A2}, 
$k_1=-3.2$, $k_2=-1.5$, $k_3=-0.6$, $k_4 = -0.3$).}
\label{fig:24soliton_2}
\end{figure}

In the usual case without elliptic background, 
no sign change occurs in the usual plane wave factor 
$e^{\xi(x,\bm{t};k)}$ where $\xi(x,\bm{t};k)$ is defined by \eqref{def:BAfcn}, 
and also in its higher order derivatives. 
This fact plays a major role in the classification of web-patterns 
\cite{ChakravartyKodama,KodamaBook}. 
However, as can be seen in Figure \ref{fig:tildePhi},
the sign change can occur for 
higher order derivatives of the real-valued Lam\'e-type plane wave factor
\eqref{real-valued_PWF}, and it may induce zeros of the $\tau$-function \eqref{transformedTauFcn}, 
which may cause singularities of $u=\partial_x^2\left(\log\tau\right)$. 
A more detailed study is needed to classify web-like patterns of elliptic solitons.
\begin{figure}[htbp]
\begin{center}
\begin{tabular}{ccc}
{\small $\tilde{\Phi}(x)$} 
& {\small $\tilde{\Phi}'(x)$} 
& {\small $\tilde{\Phi}''(x)$}\\
\includegraphics[scale=0.24]{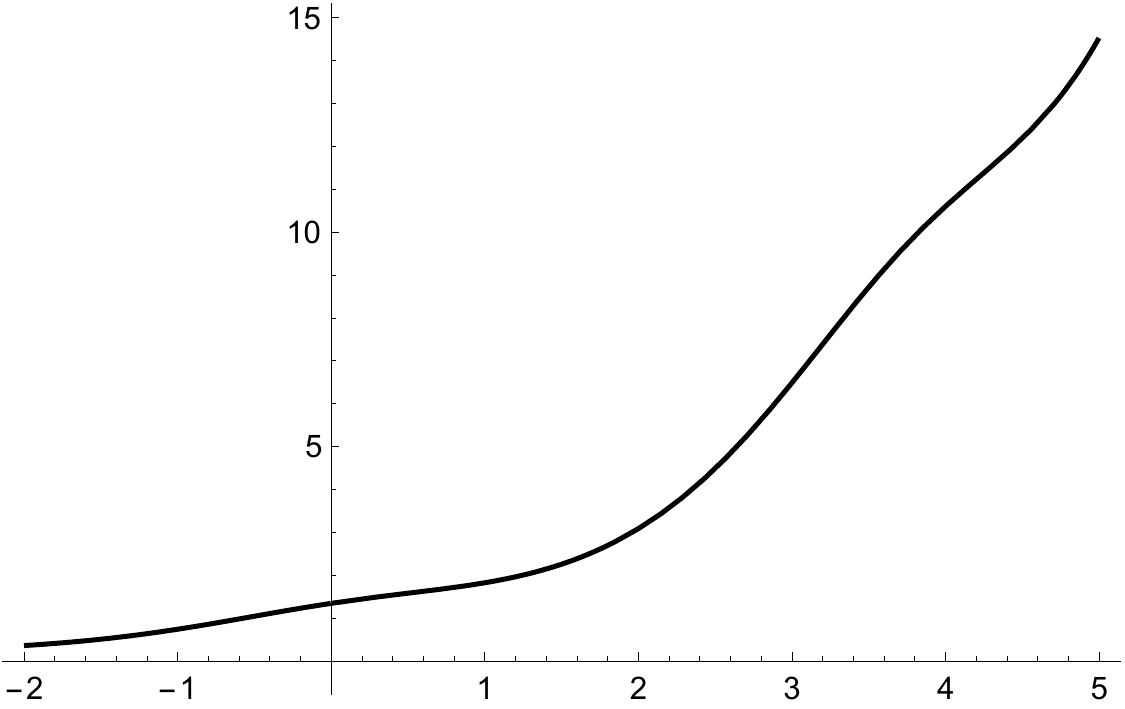}
& \includegraphics[scale=0.24]{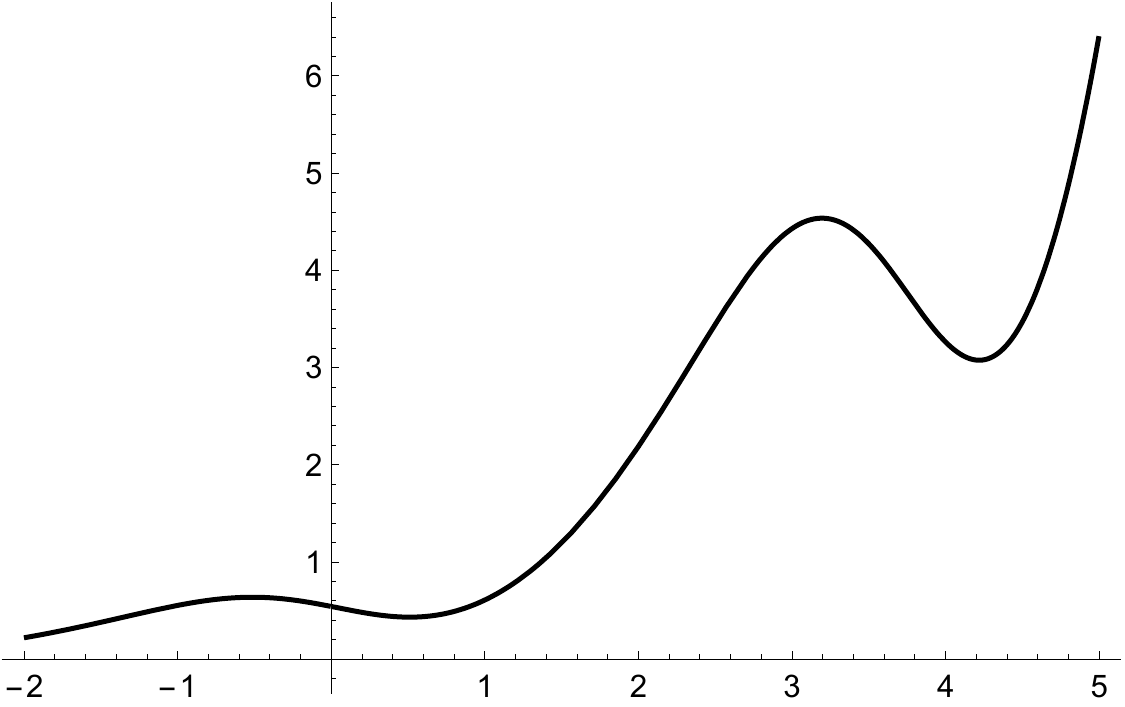}
& \includegraphics[scale=0.24]{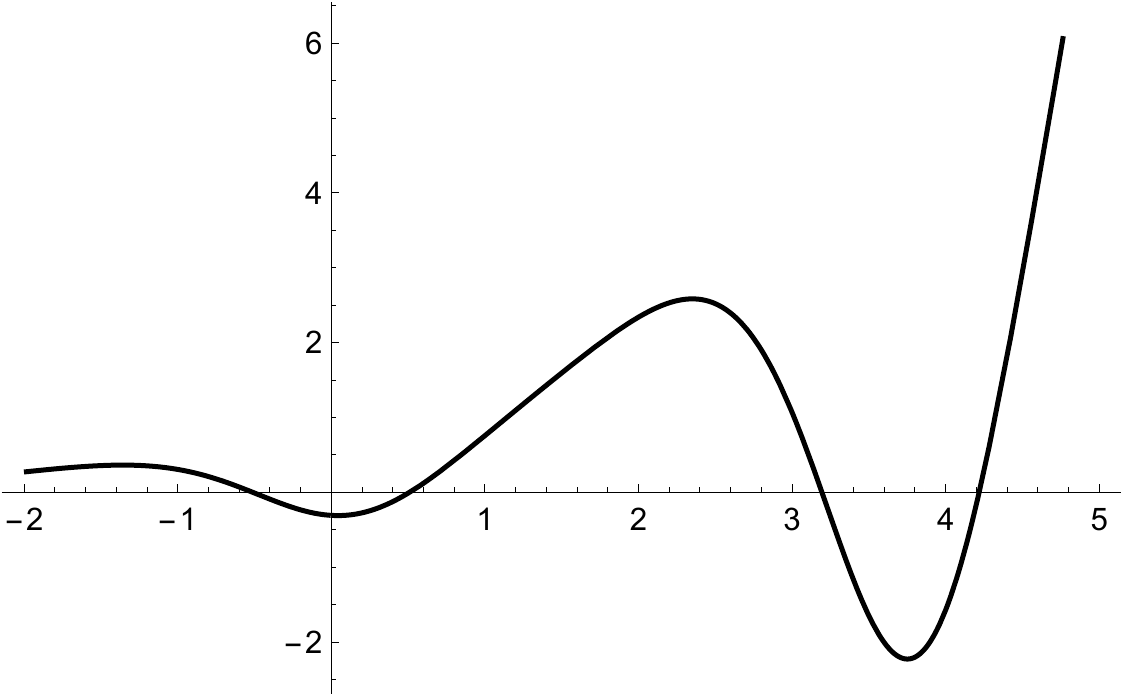}
\end{tabular}
\caption{Graphs of the function $\tilde{\Phi}$ of \eqref{real-valued_PWF} (the lemniscate case, $k=-1.2$)}
\label{fig:tildePhi}
\end{center}
\end{figure}
\newpage

\section*{Appendix}
Examples of the differential operators $P_n$, $Q_{n-1}\in\mathcal{D}_R$ ($n=2,3,\ldots$):
\begin{align*}
P_2&=\partial_x^2-2\wp(x), \quad
Q_1= \wp(x)\partial_x-\frac{\wp'(x)}{2},
\\
P_3&=\partial_x^3-3\wp(x)\partial_x-\frac{3 \wp'(x)}{2},\quad
Q_2= \wp(x)\partial_x^2-\frac{\wp'(x)}{2}\partial_x-\wp(x)^2-\frac{g_2}{4},
\\
P_4&=\partial_x^4-4 \wp(x)\partial_x^2-4\wp'(x)\partial_x+g_2-8\wp(x)^2,
\\
Q_3&= \wp(x)\partial_x^3
-\frac{\wp'(x)}{2}\partial_x^2
-2\wp(x)^2\partial_x
-\wp(x) \wp '(x),
\\
P_5&=\partial_x^5-5 \wp(x)\partial_x^3-\frac{15 \wp'(x)}{2}\partial_x^2
+3\left(g_2-10 \wp(x)^2\right)\partial_x-15 \wp(x) \wp '(x),
\\
Q_4&= \wp(x)\partial_x^4
-\frac{\wp'(x)}{2}\partial_x^3
-\left(\frac{g_2}{4}+3\wp(x)^2\right)\partial_x^2\\
& \qquad -3\wp(x)\wp'(x)\partial_x
-g_3+\frac{g_2 \wp(x)}{2}-6\wp(x)^3,
\\
P_6&=\partial_x^6-6\wp(x)\partial_x^4-12 \wp'(x)\partial_x^3
+\left(7g_2-72\wp(x)^2\right)\partial_x^2\\
&\qquad -72 \wp(x) \wp'(x)\partial_x
+2\left(-72\wp(x)^3+11 g_2 \wp(x)+8 g_3\right),
\\
Q_5&= \wp(x)\partial_x^5-\frac{\wp'(x)}{2}\partial_x^4-4\wp(x)^2\partial_x^3
-6 \wp(x) \wp '(x)\partial_x^2\\
&\qquad +\left(-2 g_3+g_2 \wp(x)-24\wp(x)^3\right)\partial_x
-\frac{1}{2}\left(24 \wp(x)^2+g_2\right)\wp '(x).
\end{align*}
Examples of the polynomials $R_n(x)$ ($n=0,1,2,\ldots$):
\begin{align*}
R_0(x) &= x, \quad \frac{R_1(x)}{3!} = x^2-\frac{g_2}{12},\quad 
\frac{R_2(x)}{5!} = x^3-\frac{3g_2}{20}x-\frac{g_3}{10},\\
\frac{R_3(x)}{7!} &= x^4-\frac{g_2}{5}x^2-\frac{g_3}{7}x+\frac{g_2^2}{560},\\
\frac{R_4(x)}{9!} &= x^5 -\frac{g_2}{4}x^3-\frac{5 g_3}{28}x^2
+\frac{g_2^2}{120}x+\frac{11 g_2 g_3}{1680},\\
\frac{R_5(x)}{11!} &= x^6-\frac{3 g_2}{10}x^4-\frac{3 g_3}{14}x^3
+\frac{7 g_2^2}{400}x^2 +\frac{57 g_2 g_3}{3080}x
-\frac{g_2^3}{26400} +\frac{g_3^2}{308}.
\end{align*}
Examples of the differential operators $\tilde{P}_n\in\mathcal{D}_R$ ($n=2,3,\ldots$):
\begin{align*}
\tilde{P}_2 &= P_2=\partial_x^2-2\wp(x),
\quad \tilde{P}_3= P_3=\partial_x^3-3\wp(x)\partial_x-\frac{3 \wp'(x)}{2},\\
\tilde{P}_4 &= P_4 -\frac{g_2}{12}=
\partial_x^4-4 \wp(x)\partial_x^2-4\wp'(x)\partial_x
+\frac{11g_2}{12}-8\wp(x)^2,\\
\tilde{P}_5 &= P_5 
=\partial_x^5-5 \wp(x)\partial_x^3-\frac{15 \wp'(x)}{2}\partial_x^2
+3\left(g_2-10 \wp(x)^2\right)\partial_x-15 \wp(x) \wp '(x),\\
\tilde{P}_6 &= P_6-\frac{3g_2}{20}P_2-\frac{g_3}{10}\\
&=\partial_6
-6 \wp(x)\partial_x^4
-12 \wp'(x)\partial_x^3
+\left\{\frac{137g_2}{20}-72 \wp(x)^2\right\}\partial_x^2\\
& \qquad\qquad -144 \wp(x)^3
-72\wp (x) \wp'(x)\partial_x
+\frac{1}{10} \left(223 g_2 \wp(x)+159 g_3\right).
\end{align*}

\end{document}